\documentclass[a4paper,UKenglish,cleveref, autoref, thm-restate,numberwithinsect]{lipics-v2021}

\pdfoutput=1 %
\hideLIPIcs  %

\title{Pathfinding in Self-Deleting Graphs} %

\author{Michal {Dvořák}}{Czech Technical University, Prague, Czech Republic}{michal.dvorak@fit.cvut.cz}{https://orcid.org/0000-0002-5048-773X}{}
\author{Dušan {Knop}}{Czech Technical University, Prague, Czech Republic}{dusan.knop@fit.cvut.cz}{https://orcid.org/0000-0003-2588-5709}{}
\author{Michal {Opler}}{Czech Technical University, Prague, Czech Republic}{michal.opler@fit.cvut.cz}{https://orcid.org/0000-0002-4389-5807}{}
\author{Jan {Pokorný}}{Czech Technical University, Prague, Czech Republic}{jan.pokorny@fit.cvut.cz}{https://orcid.org/0000-0003-3164-0791}{}
\author{Ondřej {Suchý}}{Czech Technical University, Prague, Czech Republic}{ondrej.suchy@fit.cvut.cz}{https://orcid.org/0000-0002-7236-8336}{}
\author{Krisztina {Szilágyi}}{Czech Technical University, Prague, Czech Republic}{krisztina.szilagyi@fit.cvut.cz}{https://orcid.org/0000-0003-3570-0528}{}

\authorrunning{M. Dvořák, D. Knop, M. Opler, J. Pokorný, O. Suchý, K. Szilágyi} %

\Copyright{Michal Dvořák, Dušan Knop, Michal Opler, Jan Pokorný, Ondřej Suchý, Krisztina Szilágyi} %



\category{} %

\funding{This work was co-funded by the European Union under the project Robotics and advanced industrial production (reg. no. CZ.02.01.01/00/22\_008/0004590). Michal Dvořák and Jan Pokorný acknowledge the support of Grant Agency of the Czech Technical University in Prague, grant No. SGS23/205/OHK3/3T/18.}%

\acknowledgements{}

\nolinenumbers %

\EventEditors{Ho-Lin Chen, Wing-Kai Hon, and Meng-Tsung Tsai}
\EventNoEds{3}
\EventLongTitle{36th International Symposium on Algorithms and Computation (ISAAC 2025)}
\EventShortTitle{ISAAC 2025}
\EventAcronym{ISAAC}
\EventYear{2025}
\EventDate{December 7--10, 2025}
\EventLocation{Tainan, Taiwan}
\EventLogo{}
\SeriesVolume{359}
\ArticleNo{19}
\newcommand{\itemstyle}[1]{\textcolor{lipicsGray}{\sffamily\bfseries\upshape\mathversion{bold}#1}}
\usepackage[disable]{todonotes}
\usepackage{xspace}
\usepackage{tikz-cd}
\usepackage{booktabs}

\newcommand{\NP}{{\sf NP}\xspace}
\newcommand{\coNP}{{\sf coNP}\xspace}
\newcommand{\W}[1]{{\sf W[#1]}\xspace}
\newcommand{\FPT}{{\sf FPT}\xspace}

\newcommand{\paraNP}{{\sf para-NP}\xspace}
\newcommand{\SDpath}{\textsc{Self-Deleting \mbox{$s$-$t$-path}}\xspace}
\newcommand{\SSDpath}{\textsc{Shortest Self-Deleting \mbox{$s$-$t$-path}}\xspace}
\newcommand{\SCHCpath}{\textsc{Shortest \mbox{$\chi$-compliant} \mbox{$s$-$t$-path}}\xspace}

\newcommand{\tsat}{\textsc{$3$Sat}\xspace}
\newcommand{\twsat}{\textsc{$2$Sat}\xspace}

\newcommand{\poly}{\ensuremath{\operatorname{poly}}}

\newcommand{\vcn}{\ensuremath{\operatorname{vcn}}\xspace} %
\newcommand{\td}{\ensuremath{\operatorname{td}}\xspace} %
\newcommand{\fen}{\ensuremath{\operatorname{fen}}\xspace} %
\newcommand{\fvsn}{\ensuremath{\operatorname{fvsn}}\xspace} %
\newcommand{\bw}{\ensuremath{\operatorname{bw}}\xspace} %
\newcommand{\vi}{\ensuremath{\operatorname{vi}}\xspace} %
\newcommand{\nd}{\ensuremath{\operatorname{nd}}\xspace} %
\newcommand{\cvdn}{\ensuremath{\operatorname{cvdn}}\xspace} %
\newcommand{\mw}{\ensuremath{\operatorname{mw}}\xspace} %
\newcommand{\sd}{\ensuremath{\operatorname{sd}}\xspace}

\usepackage{etoolbox}
\usepackage{placeins}
\def\ShortVersion{1} %
\ifdefined\ShortVersion
\newcommand{\sv}[1]{#1}
\newcommand{\lv}[1]{}
\newcommand{\appendixText}{}
\newcommand{\toappendix}[1]{\gappto{\appendixText}{{#1}}}
\else
\newcommand{\sv}[1]{}
\newcommand{\lv}[1]{#1}
\newcommand{\appendixText}{}
\newcommand{\toappendix}[1]{#1}
\fi
\newcommand{\appmark}{$\star$}

\newcommand{\appsection}[2]{\section{#1}\label{#2}\toappendix{\sv{\section{Omitted material from \Cref{#2}: #1}}}}
\newenvironment{apprestatable}[2]{\sv{\restatable[\appmark]{#1}{#2}}\lv{\restatable{#1}{#2}}}{\endrestatable}

\theoremstyle{plain}
\newtheorem{rrule}{Reduction Rule}
\theoremstyle{definition}
\newtheorem{construction}[theorem]{Construction}
\Crefname{rrule}{Reduction Rule}{Reduction Rules}
\Crefname{construction}{Construction}{Constructions}

\begin{document}

\maketitle

\begin{abstract}
In this paper, we study the problem of pathfinding on traversal-dependent graphs, i.e., graphs whose edges change depending on the previously visited vertices.
In particular, we study \emph{self-deleting graphs}, introduced by Carmesin et al.~\cite{carmesin2023hamiltonian}, which consist of a graph $G=(V, E)$ and a function $f\colon V\rightarrow 2^E$, where $f(v)$ is the set of edges that will be deleted after visiting the vertex $v$.
In the \textsc{(Shortest)} \SDpath problem we are given a self-deleting graph and its vertices $s$ and $t$, and we are asked to find a (shortest) path from $s$ to $t$, such that it does not traverse an edge in $f(v)$ after visiting $v$ for any vertex $v$.

We prove that \SDpath is \NP-hard even if the given graph is outerplanar, bipartite, has maximum degree $3$, bandwidth $2$ and $|f(v)|\leq 1$ for each vertex $v$.
We show that \SSDpath is \W{1}-complete parameterized by the length of the sought path and that \SDpath is \W{1}-complete parameterized by the vertex cover number, feedback vertex set number and treedepth.
We also show that the problem becomes \FPT when we parameterize by the maximum size of $f(v)$ and several structural parameters.
Lastly, we show that the problem does not admit a polynomial kernel even for parameterization by the vertex cover number and the maximum size of $f(v)$ combined already on 2-outerplanar graphs.
\end{abstract}

\section{Introduction}
Pathfinding in graphs is a well-studied topic, both from a theoretical and from a practical perspective. The famous Dijkstra's algorithm for finding a shortest path between two vertices runs in time that is quadratic in the number of vertices. However, this algorithm works under the assumption that the underlying graph is static, i.e., does not change. In many practical applications, this assumption does not hold.

One way to reflect the changes in the underlying graph is to model it as a temporal graph, where each edge is present in the graph only at certain times. In this setting, there are several ways to define the ``best'' path between two vertices: shortest path (using the smallest number of edges), fastest (the path that requires the smallest number of timesteps) and foremost (the path that requires the smallest number of timesteps starting from time zero). In all of these cases, the pathfinding problem can be solved in polynomial time~\cite{wu2016efficient}.

In some applications, however, the way that the underlying graph changes is \emph{traversal-dependent}. In other words, the availability of an edge is not time-dependent, but it depends on the previously visited vertices and edges. One such example is open-pit mining, where drilling at a vertex creates a pile of rubble which renders some edges impassable. Another example is autonomous harvesting, in which the vehicle should not return to the previously visited areas to avoid soil compactification~\cite{carmesin2023hamiltonian}. 

In this paper, we consider the model introduced by Carmesin et al.~\cite{carmesin2023hamiltonian}, called \emph{self-deleting graph}. 
A self-deleting graph is a graph $G=(V, E)$ together with a function $f\colon V\rightarrow 2^E$, which describes which edges will be deleted after visiting a vertex.
We stress that the edges in $f(v)$ are not necessarily incident to $v$. We consider the \textsc{(Shortest) Self-Deleting $s$-$t$ path} problem, where we are given a self-deleting graph and its vertices $s$ and $t$ and we are asked to find a (shortest) path from $s$ to $t$ that is valid (i.e., in which we are not traversing an edge in $f(v)$ after visiting the vertex $v$). 

\subparagraph*{Related work.} Several variants of pathfinding with restrictions on vertices and edges have been studied. 

Wojciechowski et al.~\cite{WojciechowskiSA23} introduced a problem called \textsc{Optional Choice Reachability}, where we are given a graph together with a set $S$ of pairs of edges, and we are asked to find an $s$-$t$ path that contains at most one edge from each pair in $S$. 
They showed that this problem is \NP-complete even on directed acyclic graphs (DAGs) of pathwidth 2 and \FPT parameterized by $|S|$. 

The vertex analogue of this problem, where we are given a set of pairs of vertices and we are required to use at most one of them from each pair was introduced by Krause et al.~\cite{KrauseSG1973}. 
On the one hand, the problem is \NP-hard even for DAGs~\cite{GabowMO76}, even if the set of pairs has a specific structure, such as overlapping~\cite{KolmanP09} or ordered~\cite{Kovac13}. 
On the other hand, it is polynomial time solvable, if the structure is well-parenthesized~\cite{KolmanP09}, halving~\cite{KolmanP09}, or nested~\cite{ChenKTRC01}, or in other special cases~\cite{Yinnone97}. 
Notably, Bodlaender et al.~\cite{BodlaenderJK13} studied the problem from parameterized perspective on undirected graphs, showing that it is \W{1}-hard w.r.t.\ vertex cover number of the input graph $G$, but \FPT w.r.t.\ vertex cover number of graph $H$ or treewidth of graph $G \cup H$, where graph $H$ has an edge for each forbidden pair. 
Moreover, it does not admit polynomial kernel w.r.t.\ vertex cover number of graph $G \cup H$, unless $\NP\subseteq \coNP/_{\poly}$~\cite{BodlaenderJK13}.

Szeider~\cite{Szeider03} studied the problem of finding a path where for each vertex, only certain combinations of incoming and outgoing edges are permitted. 
He provides a dichotomy between the cases which are \NP-hard and cases which are linear time solvable, based on the structure of permitted combinations. 
In the variant where the cost of each arc depends on previously traversed arcs, introduced by Kirby and Potts~\cite{KirbyP1969}, the shortest walk can be found in polynomial time~\cite{AnezBP1996,ZiliaskopoulosM1996}, but finding the shortest path is \NP-hard and hard to approximate~\cite{Wojciechowski22}.
In the variant with exclusive-disjunction arc pairs conflicts, introduced by Cerulli et al.~\cite{CerulliGSS23}, we are given pairs of arcs and we have to pay a penalty if the path contains either none of the arcs or both of them. The goal is to find a path minimizing the sum of length of edges and the penalties paid. They show that the problem is \NP-hard and provide heuristics.

\subparagraph*{Our results.} In \Cref{sec:classical_complexity}, we consider the (classical) complexity of \textsc{(Shortest)} \SDpath.
Our first result is that \SDpath is \NP-hard even on a very restricted graph class, namely outerplanar bipartite graphs of maximum degree 3. Moreover, as we show in \Cref{cor:hardness_fv1}, the result holds even when the self-deleting graph deletes at most one edge for each vertex (i.e., $|f(v)|\leq 1$ for all $v$). Next, we show a separation between \SDpath and \SSDpath: namely, on cactus graphs the former problem can be solved in linear time, whereas the latter is \NP-hard (Theorems~\ref{thm:cactus_algorithm} and~\ref{thm:hardness_cactus_shortest} respectively).

In \Cref{sec:parameterized}, we turn our attention to parameterized complexity. Firstly, we consider the parameterization of \SSDpath by solution size, and in \Cref{thm:w1_completeness_ssdpath} we show that it is $\W{1}$-complete. For most structural parameters, \SDpath turns out to be \paraNP-hard or $\W{1}$-complete. For an overview of our results on parameterizations by structural parameters, see \Cref{fig:hierarchy_structural}.

\begin{figure}[ht!]
    \centering
        \begin{tikzpicture}
        [every node/.style={
			draw=none,
			fill=gray!10,
			minimum height=2.8em,
			minimum width=4.5em,
			align=center,
			font = {\small},
            draw=black
		}]
        \tikzstyle{FPT} = [fill=green!15]
		\tikzstyle{Wc} = [fill=orange!30]
		\tikzstyle{NPh} = [fill=red!50]
		\tikzstyle{result} = [draw=black,very thick]
\def\w{1.8}
\def\h{1.3}
\node[NPh] (n_dtc) at (\w*2,-\h*0) {$\operatorname{dtc}$ \\ Cor. \ref{cor:hardness_cliques}};
\node[Wc] (n_vcn) at (\w*5,-\h*0) {$\vcn$ \\ Thm. \ref{thm:w1completeness_vc_dtlf_fvsn}};
\node[NPh] (n_nd) at (\w*3,-\h*1) {\nd \\ Cor. \ref{cor:hardness_cliques}};
\node[Wc] (n_dtlf) at (\w*6,-\h*1) {$\operatorname{dtlf}$ \\ Thm. \ref{thm:w1completeness_vc_dtlf_fvsn}};
\node[Wc] (n_vi) at (\w*5,-\h*1) {\vi \\ Thm. \ref{thm:w1completeness_vi_td}};
\node[FPT] (n_fen) at (\w*7,-\h*1) {\fen \\ Cor. \ref{cor:fpt_by_fen}};
\node[NPh] (n_gamma) at (\w*0,-\h*1) {$\gamma$ \\ Cor. \ref{cor:hardness_cliques}};
\node[NPh] (n_mim) at (\w*1,-\h*1) {$\operatorname{mim}$ \\ Cor. \ref{cor:hardness_cliques}};
\node[NPh] (n_cvdn) at (\w*4,-\h*1) {\cvdn \\ Cor. \ref{cor:hardness_cliques}};
\node[Wc] (n_fvsn) at (\w*6,-\h*2) {\fvsn \\ Thm. \ref{thm:w1completeness_vc_dtlf_fvsn}};
\node[Wc] (n_td) at (\w*5,-\h*2) {\td \\ Thm.
\ref{thm:w1completeness_vi_td}};
\node[NPh] (n_mw) at (\w*2,-\h*2) {$\operatorname{mw}$ \\ Cor. \ref{cor:hardness_cliques}};
\node[NPh] (n_dtcog) at (\w*4,-\h*2) {$\operatorname{dtcog}$ \\ Cor. \ref{cor:hardness_cliques}};
\node[NPh] (n_diam) at (\w*2,-\h*3) {$\operatorname{diam}$ \\ Cor. \ref{cor:hardness_cliques}};

\node[NPh] (n_sd) at (\w*3,-\h*3) {$\sd$ \\ Cor. \ref{cor:hardness_cliques}};

\draw[->] (n_sd) to (n_nd);
\draw[->] (n_sd) to[in=250,out=25] (n_td);
\draw[->] (n_sd) to[in=210,out=35] (n_cvdn);

\draw[->] (n_nd) to (n_dtc);
\draw[->] (n_nd) to[out=30,in=200] (n_vcn);
\draw[->] (n_dtlf) to (n_vcn);
\draw[->] (n_vi) to (n_vcn);
\draw[->] (n_cvdn) to[out=150,in=340] (n_dtc);
\draw[->] (n_fvsn) to (n_fen);
\draw[->] (n_fvsn) to (n_dtlf);
\draw[->] (n_td) to (n_vi);
\draw[->] (n_mw) to (n_nd);
\draw[->] (n_dtcog) to (n_cvdn);

\draw[->] (n_diam) to[out=140,in=300] (n_gamma);
\draw[->] (n_diam) to[out=140,in=270] (n_mim);
\draw[->] (n_diam) to[in=205,out=35] (n_td);
\draw[->] (n_diam) to[out=35,in=200] (n_dtcog);
\draw[->] (n_diam) to (n_mw);

\draw[->] (n_gamma) to[out=60,in=195] (n_dtc);
\draw[->] (n_mim)[out=50,in=200] to (n_dtc);
\draw[->] (n_cvdn) to (n_vcn);

\end{tikzpicture}
    \caption{Results for \SDpath parameterized by structural parameters. Red color stands for \paraNP-hard, green is \FPT and orange is \W{1}-complete. Each parameter is accompanied by the corresponding statement from which the result follows. An arrow $\alpha\to\beta$ indicates a functional upper bound, i.e., $\alpha \leq g(\beta)$ for some (computable) function~$g$. For full names of the parameters refer to \Cref{sec:preliminaries}.}
    \label{fig:hierarchy_structural}
\end{figure}
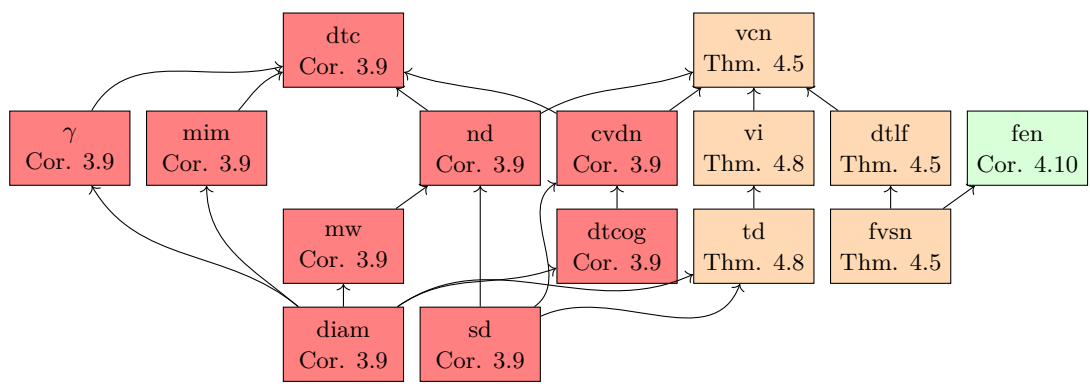

In order to make the problem tractable, we consider the case of bounded deletion set size (i.e. we parameterize the problem by $\mu=\max_{v\in V}|f(v)|$). Although this parameterization alone does not make the problem tractable (in particular, we show that \SSDpath is \paraNP-hard parameterized by $\mu$), it turns out that the parameterization by $\mu$ and $k$ (the number of vertices of the sought path) leads to an \FPT algorithm. This result allows us to obtain several \FPT algorithms for parameterizations by $\mu$ and structural parameters, such as vertex cover number, vertex integrity and treedepth. Using the results of Dvořák et al.~\cite{DvorakKOPSS2025} on the number of edges in traceable graphs with dense structure, we obtain \FPT algorithms for several dense parameters and $\mu$, such as neighborhood diversity, shrub-depth, modular-width, and size of maximum induced matching. For an overview of our results about parameterizations by structural parameters and $\mu$, refer to \Cref{fig:hierarchy_structural_and_mu}.

\begin{figure}[ht!]
    \centering
    \begin{tikzpicture}
    [every node/.style={
			draw=none,
			fill=gray!10,
			minimum height=2.8em,
			minimum width=4.5em,
			align=center,
			font = {\small},
            draw=black
		}]
        \tikzstyle{FPT} = [fill=green!15]
		\tikzstyle{Wc} = [fill=orange!30]
		\tikzstyle{NPh} = [fill=red!50]
		\tikzstyle{result} = [draw=black,very thick]
\def\w{1.8}
\def\h{1.3}
\node[FPT] (n_dtc) at (\w*2,-\h*0) {$\operatorname{dtc}$ \\ Thm. \ref{thm:fpt_by_dense_and_mu}};
\node[FPT] (n_vcn) at (\w*5,-\h*0) {$\vcn$ \\ Cor. \ref{cor:fpt_by_mu_and_vc_vi_td}};
\node[FPT] (n_nd) at (\w*3,-\h*1) {\nd \\ Thm. \ref{thm:fpt_by_dense_and_mu}};
\node[Wc] (n_dtlf) at (\w*6,-\h*1) {$\operatorname{dtlf}$ \\ Thm. \ref{thm:w1completeness_dtlf_mu1}};
\node[FPT] (n_vi) at (\w*5,-\h*1) {\vi \\ Cor. \ref{cor:fpt_by_mu_and_vc_vi_td}};
\node[FPT] (n_fen) at (\w*7,-\h*1) {\fen \\ Cor. \ref{cor:fpt_by_fen}};
\node[NPh] (n_gamma) at (\w*0,-\h*1) {$\gamma$ \\ Thm. \ref{thm:hardness_domination}};
\node[FPT] (n_mim) at (\w*1,-\h*1) {$\operatorname{mim}$ \\ Thm. \ref{thm:fpt_by_dense_and_mu}};
\node[FPT] (n_cvdn) at (\w*4,-\h*1) {\cvdn \\ Thm. \ref{thm:fpt_by_dense_and_mu}};
\node[Wc] (n_fvsn) at (\w*6,-\h*2) {\fvsn \\ Thm. \ref{thm:w1completeness_dtlf_mu1}};
\node[FPT] (n_td) at (\w*5,-\h*2) {\td \\ Cor. \ref{cor:fpt_by_mu_and_vc_vi_td}};
\node[FPT] (n_mw) at (\w*2,-\h*2) {$\operatorname{mw}$ \\ Thm. \ref{thm:fpt_by_dense_and_mu}};
\node[FPT] (n_dtcog) at (\w*4,-\h*2) {$\operatorname{dtcog}$ \\ Thm. \ref{thm:fpt_by_dense_and_mu}};
\node[NPh] (n_diam) at (\w*2,-\h*3) {$\operatorname{diam}$ \\ Thm. \ref{thm:hardness_domination}};

\node[FPT] (n_sd) at (\w*3,-\h*3) {$\sd$ \\ Thm. \ref{thm:fpt_by_dense_and_mu}};

\draw[->] (n_sd) to (n_nd);
\draw[->] (n_sd) to[in=250,out=25] (n_td);
\draw[->] (n_sd) to[in=210,out=35] (n_cvdn);

\draw[->] (n_nd) to (n_dtc);
\draw[->] (n_nd) to[out=30,in=200] (n_vcn);
\draw[->] (n_dtlf) to (n_vcn);
\draw[->] (n_vi) to (n_vcn);
\draw[->] (n_cvdn) to[out=150,in=340] (n_dtc);
\draw[->] (n_fvsn) to (n_fen);
\draw[->] (n_fvsn) to (n_dtlf);
\draw[->] (n_td) to (n_vi);
\draw[->] (n_mw) to (n_nd);
\draw[->] (n_dtcog) to (n_cvdn);

\draw[->] (n_diam) to[out=140,in=300] (n_gamma);
\draw[->] (n_diam) to[out=140,in=270] (n_mim);
\draw[->] (n_diam) to[in=205,out=35] (n_td);
\draw[->] (n_diam) to[out=35,in=200] (n_dtcog);
\draw[->] (n_diam) to (n_mw);

\draw[->] (n_gamma) to[out=60,in=195] (n_dtc);
\draw[->] (n_mim)[out=50,in=200] to (n_dtc);
\draw[->] (n_cvdn) to (n_vcn);

\end{tikzpicture}
    \caption{Results for \SDpath parameterized by structural parameters and $\mu$ combined. The meaning of colors and arrows is the same as in \Cref{fig:hierarchy_structural}.}
    \label{fig:hierarchy_structural_and_mu}
\end{figure}
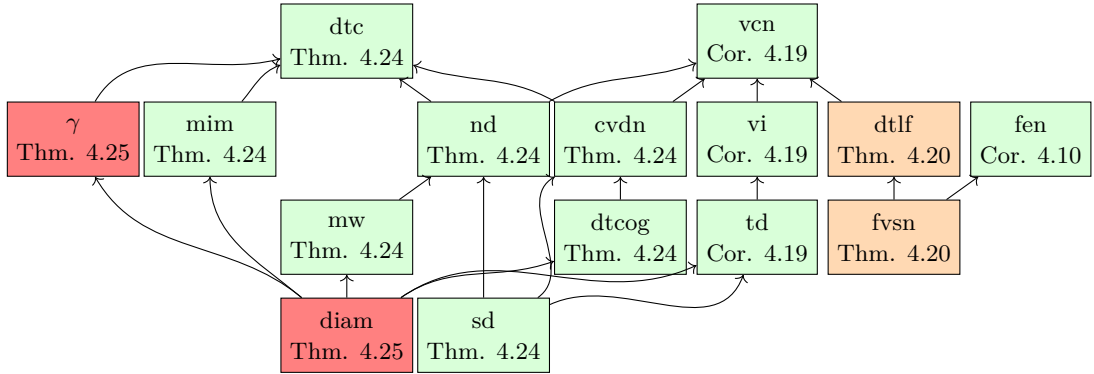

Lastly, in \Cref{sec:kernel}, we consider kernelizing the problem. We show that, under standard parameterized complexity assumptions, \SDpath does not admit a polynomial kernel parameterized by $\mu$ and vertex cover number even on $2$-outerplanar graphs. On the positive side, we obtain a linear kernel for the parameterization by feedback edge number and linear kernel for the parameterization by vertex cover number on outerplanar graphs. Moreover, while \SDpath does not admit a classical kernel parameterized by $\mu$ on cliques, we obtain a linear Turing kernel on cliques parameterized by $\mu$.

\section{Preliminaries}
\label{sec:preliminaries}
\sv{
\toappendix{\section{Additional Preliminaries}}
}
For an integer $k$ we let $[k]=\{1,2,\ldots,k\}$. For any function $f\colon A\to B$ and $X\subseteq A$, the restriction of $f$ to $X$ is denoted $f|_X$.
We use the notation $\widetilde{O}(\cdot)$ to suppress polylogarithmic factors in time complexity.

\subparagraph*{Graph Theory.}
We use standard notations, terminology, and definitions of graph theory, refer to Diestel~\cite{Diestel2017GT} for those. Unless explicitly stated, we consider simple undirected graphs. We denote a $v_1$–$v_k$ walk (or path) in a graph by the sequence
$(v_1,e_1,v_2,\ldots,e_{k-1},v_k)$ of vertices $v_i$ and edges $e_i$, where $e_i=\{v_i,v_{i+1}\}$. For brevity, we may sometimes omit the edges and write just the sequence of vertices. We denote by $\operatorname{cc}(G)$ the number of connected components of a graph $G$. A graph is a \emph{cactus} if every edge lies on at most one cycle. A graph is a \emph{block graph} if every vertex-$2$-connected component is a clique. A $2\times \ell$ grid for some positive integer $\ell$ is a \emph{ladder}. The class of \emph{cographs} is the minimal class of graphs containing the one-vertex graph and closed under complete joins and disjoint unions of two graphs.

\subparagraph*{Self-deleting graphs.}
A \emph{self-deleting graph} is an ordered pair $(G,f)$, where $G=(V,E)$ is an undirected graph and $f$ is a function \mbox{$f\colon V\to 2^E$}. We write $G$ instead of $(G,f)$ if $f$ is clear from the context. The function $f$ is called the \emph{deletion function}. Let $(G,f)$ be a self-deleting graph. A path or walk $(v_1,e_1,v_2,e_2,\ldots,e_{k-1},v_k)$ in $G$ is \emph{$f$-conforming} if for every $j\leq i$ we have $e_i\notin f(v_j)$. We denote $\mu_f=\max_{v\in V}|f(v)|$ and $|f|=\sum_{v\in V}|f(v)|$. If $f$ is clear from the context, we omit the lower index and write just $\mu$.
Central to our paper are the two decision problems \SDpath and \SSDpath. In \SDpath we are given a self-deleting graph $(G,f)$ and two vertices $s,t\in V(G)$ and the task is to decide if there is an $f$-conforming $s$-$t$ path in $G$. In \SSDpath we are in addition given a positive integer $k$ and the task is to decide if there is an $f$-conforming $s$-$t$ path in $G$ on at most $k$ vertices.

We assume that $\mu \geq 1$ as otherwise the problems are trivial. Moreover, if $\mu$ appears in the base of an exponential or in a logarithm, $\mu$ should be replaced by $\max\{\mu,2\}$ in order to avoid degenerate cases. We stick to writing $\mu$ for the sake of readability.  

\subparagraph*{Parameterized complexity.}
We assume the reader is familiar with parameterized complexity. For definitions and comprehensive overview, refer to the monograph of Cygan et al.~\cite{CyganParaAlg2015}. In our work we consider the following structural parameters: \emph{bandwidth} ($\bw$), \emph{distance to clique} ($\operatorname{dtc}$), \emph{vertex cover number} ($\vcn$), \emph{domination number} ($\gamma$), \emph{maximum induced matching} ($\operatorname{mim}$), \emph{neighborhood diversity} ($\nd$), \emph{diameter} ($\operatorname{diam}$), \emph{cluster vertex deletion number} ($\cvdn$), \emph{feedback vertex set number} ($\fvsn$), \emph{distance to linear forest} ($\operatorname{dtlf}$),  \emph{vertex integrity} ($\vi$), \emph{feedback edge number} (\fen), \emph{treedepth} ($\td$), \emph{distance to cograph} ($\operatorname{dtcog}$), \emph{modular-width} ($\operatorname{mw}$), and \emph{shrub-depth} ($\sd$).
\sv{
Formal definitions of all the parameters can be found in~\Cref{subsec:definitions_of_parameters}.
}

\toappendix{

\subsection{Definitions of Parameters}\label{subsec:definitions_of_parameters}
Let $G=(V,E)$ be a graph. The \emph{bandwidth} of $G$ is defined as \[\bw(G)=\min\Big\{\max\big\{|\iota(u)-\iota(v)| \,\big|\, \{u,v\}\in E \big\} \,\Big|\, \iota\colon V \to \mathbb{N} \text{ injective}\Big\}.\]

The \emph{vertex integrity} of a graph $G$, denoted $\vi(G)$, is the smallest number $k$ such that there is a set of at most $k$ vertices whose removal results in a graph where each connected component is of size at most $k$. 

The \emph{neighborhood diversity} of a graph $G$, denoted $\nd(G)$, is the smallest integer $w$ such that there exists a partition of $V$ into $w$ sets $V_1,\ldots,V_w$ such that for any $i\in[w]$ and $u,v\in V_i$ it holds $N(u)\setminus \{v\}=N(v)\setminus \{u\}$. In other words, the sets $V_i$ are either independent or cliques and for any two distinct $V_i,V_j$ either we have for every $v_i\in V_i,v_j\in V_j$ that $\{v_i,v_j\}\in E(G)$ or for every $v_i\in V_i,v_j\in V_j$ that $\{v_i,v_j\}\notin E(G)$.

The following definition is from~\cite{GajarskyLO2013}.
Consider an algebraic expression $A$ that uses the following operations:
\begin{enumerate}[(O1)]
    \item create an isolated vertex;\label{itm:modular_width_O1}
    \item take the \emph{disjoint union} of graphs $G_1,G_2$, denoted by $G_1\oplus G_2$, which is the graph with vertex set $V(G_1)\cup V(G_2)$ and edge set $E(G_1)\cup E(G_2)$;\label{itm:modular_width_O2}
    \item take the \emph{complete join} of $2$ graphs $G_1$ and $G_2$, denoted by $G_1\otimes G_2$, which is the graph with vertex set $V(G_1)\cup V(G_2)$ and edge set $E(G_1)\cup E(G_2)\cup \{\{v,w\}\mid v\in V(G_1)\wedge w\in V(G_2)\}$;\label{itm:modular_width_O3}
    \item for graphs $G_1,\ldots, G_n$ and a pattern graph $G$ with vertices $v_1,\ldots,v_n$ perform the \emph{substitution} of the vertices of $G$ by the graphs $G_1,\ldots, G_n$, denoted by $G(G_1,\ldots, G_n)$, which is the graph with vertex set $\bigcup_{i=1}^n V(G_i)$ and edge set $\bigcup_{i=1}^n E(G_i)\cup \{\{u,v\}\mid u \in V(G_i)\wedge v \in V(G_j)\wedge \{v_i,v_j\}\in E(G)\}$. Hence, $G(G_1,\ldots, G_n)$ is obtained from $G$ by replacing every vertex $v_i\in V(G)$ with the graph $G_i$ and adding all edges between vertices of a graph $G_i$ and the vertices of a graph $G_j$ whenever $\{v_i,v_j\}\in E(G)$.\label{itm:modular_width_O4}
\end{enumerate}
\newcommand{\oitemref}[1]{\itemstyle{\hyperref[#1]{(O\ref{#1})}}}
The \emph{width} of the expression $A$ is the maximum number of vertices of a pattern graph used by any occurrence of the operation \oitemref{itm:modular_width_O4} in $A$ (or $0$ if \oitemref{itm:modular_width_O4} does not occur in $A$). The \emph{modular-width} of a graph~$G$, denoted $\operatorname{mw}(G)$, is the smallest integer $m$ such that $G$ can be obtained from such an algebraic expression of width at most $m$. Note that the operations \oitemref{itm:modular_width_O2} and \oitemref{itm:modular_width_O3} can be seen as a special case of \oitemref{itm:modular_width_O4} with graphs $K_2$, resp. $\overline{K_2}$. 
A graph $G$ is a \emph{cograph} if it has modular width $0$.

Let $\Pi$ be a graph property. A set $S\subseteq V$ (resp. $F\subseteq E$) is a \emph{vertex-modulator (resp. edge-modulator) to $\Pi$} if $G\setminus S\in \Pi$ (resp. $G\setminus F\in \Pi$). The parameter \emph{vertex-distance to $\Pi$} (resp. \emph{edge-distance to $\Pi$}) is the size of the smallest vertex-modulator (resp. edge-modulator) to $\Pi$. The \emph{vertex cover number} is the vertex-distance to edgeless graph. The \emph{feedback vertex set number} (\fvsn) is the vertex-distance to forest. The \emph{feedback edge number} (\fen) is the edge-distance to forest. A \emph{cluster} (graph) is a disjoint union of cliques. The \emph{cluster vertex deletion number} (\cvdn) is the distance to cluster graph. 

A \emph{linear forest} is a forest where each connected component is a path. The \emph{treedepth} of $G$, denoted $\td(G)$, is defined to be the smallest possible depth of a rooted forest $F$ with $V(F)\supseteq V(G)$ such that every edge of $G$ is in ancestor-descendant relationship in $F$.
A set $D\subseteq V(G)$ is \emph{dominating set of $G$} if each vertex of $G$ is in $D$ or has a neighbor in $D$. The \emph{domination number} of $G$, denoted $\gamma(G)$, is the size of smallest dominating set of $G$.

\subparagraph*{Shrub-depth}
\begin{definition}[Tree-model~\cite{Ganian2017ShrubdepthCH}]
    Let $m$ and $d$ be non-negative integers. A \emph{tree-model of $m$ colours and depth $d$} for a graph $G$ is a pair $(T,S)$ of a rooted tree $T$ (of height $d$) and a set $S\subseteq [m]^2\times [d]$ (called a \emph{signature} of the tree-model) such that
    \begin{enumerate}
        \item the length of each root-to-leaf path in $T$ is exactly $d$,
        \item the set of leaves of $T$ is exactly the set $V(G)$ of vertices of $G$,
        \item each leaf of $T$ is assigned one of the colours in $[m]$, and
        \item for any $i,j,\ell$ it holds that $(i,j,\ell)\in S\Leftrightarrow (j,i,\ell)\in S$ (symmetry in the colours), and 
        \item for any two vertices $u,v\in V(G)$ and any $i,j,\ell$ such that $u$ is coloured $i$ and $v$ is coloured $j$ and the distance between $u,v$ in $T$ is $2\ell$, the edge $\{u,v\}$ exists in $G$ if and only if $(i,j,\ell)\in S$.
    \end{enumerate}
\end{definition}
\begin{definition}[Shrub-depth]
    A class $\mathcal{G}$ of graphs has \emph{shrub-depth} at most $d$ if there exists~$m$ such that each $G\in\mathcal{G}$ admits a $(d,m)$ tree-model.
\end{definition}

}%

\subparagraph*{Exponential Time Hypothesis.} Exponential Time Hypothesis (ETH), introduced by Impagliazzo, Paturi and Zane~\cite{ImpagliazzoP01,ImpagliazzoPZ01journal} asserts, roughly speaking, that there is no algorithm for \tsat in time $2^{o(n')}$, where $n'$ is the number of variables of the input formula. In fact, even $2^{o(n'+m')}$ algorithm is ruled out by using the Sparsification Lemma~\cite{ImpagliazzoPZ01journal}, where $m'$ is the number of clauses of the input formula. We also utilize the result of Chen, Huang, Kanj, and Xia~\cite{ChenHKX04conference,ChenHKX06journal} that there is no $g(k)n^{o(k)}$ algorithm for \textsc{(Multicolored) Clique} for any computable function~$g$ unless ETH fails.

\sv{
Statements where proofs or details are omitted due to space constraints are marked with~\appmark. The omitted material is available in the appendix.
}

\appsection{Classical Complexity}{sec:classical_complexity}
As observed by Carmesin et al.~\cite[Lemma 3]{carmesin2023hamiltonian}, if every vertex deletes only its incident edges, \SDpath reduces to pathfinding in a directed graph: deleting $\{u,v\}$ at $v$ orients the edge from $u$ to $v$, and if both endpoints delete it, the edge is removed entirely. Thus, \SDpath is solvable in linear time when all $f(v)$ contain only edges incident to $v$. 
Another scenario when the problem is tractable is that the graph has only a constant number of $s$-$t$ paths, e.g., in trees or graphs of maximum degree $2$. The problem becomes hard already on graphs of maximum degree $3$. We show a polynomial-time reduction from \tsat to \SDpath which we also utilize later with slight modifications.

\begin{figure}[ht!]
    \centering
    \begin{tikzpicture}[scale=.85]
    \def\w{0.8}
    \def\h{0.8}
    \node (ix) at (\w*0,\h*0) [circle,draw,fill=black,label=below left:${s=\iota_x}$] {};
    \node (ox) at (\w*2,\h*0) [circle,draw,fill=black,label=above:$o_x$] {};
    \node (Tx) at (\w*1,\h*1) [circle,draw,fill=red,label=above:$T_x$] {};
    \node (Fx) at (\w*1,\h*-1) [circle,draw,fill=blue,label=below:$F_x$] {};

    \draw (ix) -- (Tx) -- (ox) -- (Fx) -- (ix);

    \node (iy) at (\w*3,\h*0) [circle,draw,fill=black,label=below:$\iota_y$] {};
    \node (oy) at (\w*5,\h*0) [circle,draw,fill=black,label=above:$o_y$] {};
    \node (Ty) at (\w*4,\h*1) [circle,draw,fill=green,label=above:$T_y$] {};
    \node (Fy) at (\w*4,\h*-1) [circle,draw,fill=orange,label=below:$F_y$] {};

    \draw (iy) -- (Ty) -- (oy) -- (Fy) -- (iy);

    \node (iz) at (\w*6,\h*0) [circle,draw,fill=black,label=below:$\iota_z$] {};
    \node (oz) at (\w*8,\h*0) [circle,draw,fill=black,label=above:$o_z$] {};
    \node (Tz) at (\w*7,\h*1) [circle,draw,fill=violet,label=above:$T_z$] {};
    \node (Fz) at (\w*7,\h*-1) [circle,draw,fill=pink,label=below:$F_z$] {};

    \draw (iz) -- (Tz) -- (oz) -- (Fz) -- (iz);

    \draw (ox) to (iy);
    \draw (oy) to (iz);

\begin{scope}[xshift=8cm,yshift=2cm]
    \node (A1) at (\w*-1,\h*-2) [circle,draw,fill=black,label=above:$\iota_{C_1}$] {};
    \node (B1) at (\w*0,\h*-2) [circle,draw,fill=black] {};
    \node (C1) at (\w*1,\h*-2) [circle,draw,fill=black] {};
    
    \node (D1) at (\w*-1,\h*-3) [circle,draw,fill=black] {};
    \node (E1) at (\w*0,\h*-3) [circle,draw,fill=black] {};
    \node (F1) at (\w*1,\h*-3) [circle,draw,fill=black,label=below:$o_{C_1}$] {};

    \draw (A1) -- (B1);
    \draw (B1) -- (C1);
    \draw (D1) -- (E1);
    \draw (E1) -- (F1);

    \draw[blue,ultra thick] (A1) -- node[left]{} (D1);
    \draw[green,ultra thick] (B1) -- node[left]{} (E1);
    \draw[violet,ultra thick] (C1) -- node[right]{} (F1);

    \node (A2) at (\w*2,\h*-2) [circle,draw,fill=black,label=above:$\iota_{C_2}$] {};
    \node (B2) at (\w*3,\h*-2) [circle,draw,fill=black] {};
    \node (C2) at (\w*4,\h*-2) [circle,draw,fill=black] {};
    
    \node (D2) at (\w*2,\h*-3) [circle,draw,fill=black] {};
    \node (E2) at (\w*3,\h*-3) [circle,draw,fill=black] {};
    \node (F2) at (\w*4,\h*-3) [circle,draw,fill=black,label=below:$o_{C_2}$] {};

    \draw (A2) -- (B2);
    \draw (B2) -- (C2);
    \draw (D2) -- (E2);
    \draw (E2) -- (F2);

    \draw[red,ultra thick] (A2) -- node[left]{} (D2);
    \draw[orange,ultra thick] (B2) -- node[left]{} (E2);
    \draw[violet,ultra thick] (C2) -- node[right]{} (F2);

    \node (A3) at (\w*5,\h*-2) [circle,draw,fill=black,label=above:$\iota_{C_3}$] {};
    \node (B3) at (\w*6,\h*-2) [circle,draw,fill=black] {};
    \node (C3) at (\w*7,\h*-2) [circle,draw,fill=black] {};
    
    \node (D3) at (\w*5,\h*-3) [circle,draw,fill=black] {};
    \node (E3) at (\w*6,\h*-3) [circle,draw,fill=black] {};
    \node (F3) at (\w*7,\h*-3) [circle,draw,fill=black,label=below right:${t=o_{C_3}}$] {};

    \draw (A3) -- (B3);
    \draw (B3) -- (C3);
    \draw (D3) -- (E3);
    \draw (E3) -- (F3);

    \draw[red,ultra thick] (A3) -- node[]{} (D3);
    \draw[orange,ultra thick] (B3) -- node[left]{} (E3);
    \draw[pink,ultra thick] (C3) -- node[right]{} (F3);

    \draw(F1) -- (A2);
    \draw(F2) -- (A3);
\end{scope}
    \draw (oz) to (A1);

    \end{tikzpicture}
    \caption{Example of \Cref{construction:sat} for the input formula $\varphi=(x\vee \neg y \vee\neg z)\wedge (\neg x \vee y \vee \neg z) \wedge (\neg x \vee y \vee z)$. The deletion sets are indicated by the colors. 
    }
    \label{fig:construction_sat_example}
\end{figure}
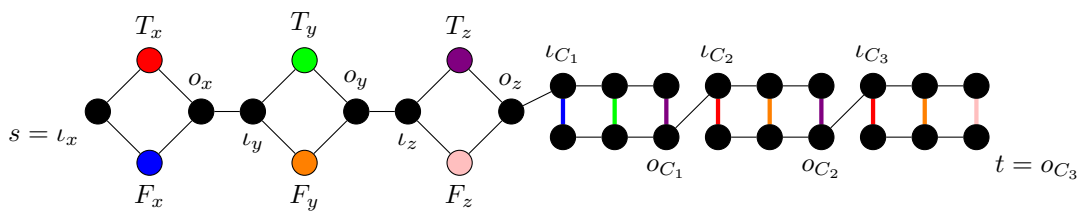

\begin{construction}\label{construction:sat}
Let $\varphi$ be a given CNF-formula over variables $X=\{x_1,\ldots,x_{n'}\}$ with clauses $\mathcal{C}=\{C_1,\ldots C_{m'}\}$. Our construction of an instance of \SDpath consist of variable and clause gadgets that we now introduce. See \Cref{fig:construction_sat_example} for an example. 
\begin{description}
    \item[Variable gadget] The variable gadget for a variable $x\in X$ is a $4$-cycle with vertices $T_x,\iota_x,F_x,o_x$. %
    \item[Clause gadget] The clause gadget for a clause $C\in\mathcal{C}$ is the $2\times |C|$ grid. The top left corner vertex is denoted $\iota_C$ (input) and the bottom right corner vertex is denoted $o_C$ (output). Each of the $|C|$ vertical edges correspond to a literal of $C$. For a literal $\ell \in C$ we denote its corresponding edge by $e^C_\ell$.
\end{description}

Create one variable gadget for each variable and one clause gadget for each clause. Add edges of the form $\{o_{x_k},\iota_{x_{k+1}}\}$ for every $k\in[n'-1]$, add edge $\{o_{x_{n'}},\iota_{C_1}\}$ and finally the edges $\{o_{C_i},\iota_{C_{i+1}}\}$ for every $i\in[m'-1]$. We let $s=\iota_{x_1}$ and $t=o_{C_{m'}}$. It remains to specify the deletion function~$f$. For all vertices $v$ except $T_x,F_x$ inside the variable gadgets, we set $f(v)=\emptyset$. For a variable $x\in X$ we set $f(T_x)=\{e_{\neg x}^C\mid C\in \mathcal{C},\neg x\in C\}$ and $f(F_x)=\{e_x^C\mid C\in \mathcal{C},x \in C\}$. In other words, $T_x$ deletes the literal edges in clause gadgets corresponding to literals $\neg x$ and $F_x$ deletes the literal edges corresponding to literals $x$.
The resulting \SDpath instance is $(G,f,s,t)$.
\end{construction}

\begin{apprestatable}{lemma}{lemmaConstructionSatCorrectness}\label{claim:construction_sat_correctness}
    Let $\varphi$ be the formula and $(G,f,s,t)$ the \SDpath instance obtained by \Cref{construction:sat} from~$\varphi$. Then $\varphi$ is satisfiable if and only if there is an $f$-conforming $s$-$t$ path in $G$.
\end{apprestatable}

\toappendix{
\sv{\lemmaConstructionSatCorrectness*}
\begin{claimproof}
    $\Rightarrow$: Let $\pi\colon X\to\{0,1\}$ be a satisfying assignment for $\varphi$. 
    Consider a path that starts at $s=\iota_{x_i}$ and at each variable gadget either visits the vertex $T_x$ if $\pi(x)=1$ or $F_x$ if $\pi(x)=0$, and then the vertex $o_x$.
    It end in vertex $o_{x_n}$ and takes the edge to $\iota_{C_1}$. 
    Now, for each clause $C$ there is a literal $\ell\in C$ that is satisfied by $\pi$. 
    By construction, the edge $e_{\ell}^C$ is not deleted by traversing the variable gadgets, hence we use it to pass from $\iota_{C}$ to $o_C$. 
    This happens for every $C$ and eventually we arrive at $o_{C_m}=t$.
    As no edge was deleted, the path is $f$-conforming.
    
    $\Leftarrow$: 
    Let $P$ be an $f$-conforming $s$-$t$ path in $G$.
    By construction, any path from $\iota_{x_1}$ to $o_{x_n}$ uses exactly one of $T_x$ or $F_x$ for each variable $x\in X$. 
    This gives rise to an assignment of variables. 
    For each clause $C$ an edge $e_{\ell_i}^C$ among $e_{\ell_1}^C,\ldots,e_{\ell_{|C|}}^C$ is used by $P$.
    It follows that the literal $\ell_i$ is satisfied by the assignment, as otherwise the edge would be deleted.
    Thus the assignment satisfies all the clauses and $\varphi$ is satisfiable.
    But since there is no satisfying assignment for $\varphi$, there exists a clause $C$ for which all the edges $e_{\ell_1}^C,\ldots,e_{\ell_{|C|}}^C$ are deleted. Since those form an $s$-$t$ cut, there cannot be an $f$-conforming $s$-$t$ path in $(G,f)$.
\end{claimproof}
}%

By using \Cref{construction:sat} together with \Cref{claim:construction_sat_correctness} we immediately obtain the following theorem.

\begin{theorem}
    \SDpath is \NP-hard.
\end{theorem}

\begin{corollary}\label{cor:hardness_fv1}
    \SDpath is \NP-hard even if the deletion function $f$ satisfies $\forall v \in V:|f(v)|\leq 1$, i.e., $\mu \leq 1$.
\end{corollary}
\begin{proof}
    Use \Cref{construction:sat} with the following modification. Replace the vertex $T_x$ (resp. $F_x$) by a path on $|f(T_x)|$ (resp. $|f(F_x)|$) vertices and delete one edge of the original set $f(v)$ on each vertex of the path to obtain $|f(v)|\leq 1$.
\end{proof}

\begin{remark}\label{rem:sat_is_linear}
    Note that \Cref{construction:sat} produces a graph with $O(n'+m')$ vertices and edges, where $n'$ and $m'$ are the number of variables and clauses in the original \tsat formula.
\end{remark}

\begin{apprestatable}{lemma}{lemmaStructuralPropertiesGrid}
\label{lem:structural_properties_grid}
    $2\times \ell$ grid has bandwidth $2$, treedepth at most $O(\log \ell)$ and vertex integrity at most $O(\sqrt{\ell})$.
\end{apprestatable}
\toappendix{
\sv{\lemmaStructuralPropertiesGrid*}
\begin{proof}
\begin{figure}[ht!]
        \centering
        \begin{tikzpicture}[scale=1, every node/.style={circle, draw, inner sep=2pt}]
    \node (a) at (0,0){};
    \node (b) at (1,0){};
    \node (c) at (2,0){};
    \node (d) at (3,0){};
    \node (e) at (4,0){};
    \node (f) at (5,0){};
    \node (g) at (6,0){};
    \node (h) at (7,0){};

    \draw (a)--(b);
    \draw (c)--(d);
    \draw (e)--(f);
    \draw (g)--(h);
    
    \draw[bend left] (a) to (c);
    \draw[bend right] (b) to (d);
    \draw[bend left] (c) to (e);
    \draw[bend right] (d) to (f);
    \draw[bend right] (f) to (h);
    \draw[bend left] (e) to (g);

    \node[draw=none] (dots) at (8,0){$\cdots$};
    \end{tikzpicture}
        \caption{Linear layout of bandwidth $2$ for $2\times \ell$ grid. The length-one edges resemble the vertical edges of the grid and the length-two edges are the horizontal ones. In general, the $i$-th vertex is connected to the $(i-2)$-th and $(i+2)$-th. Moreover, if $i$ is even, then it is connected to the $(i-1)$-th, if $i$ is odd, it is connected to the $(i+1)$-th.}
        \label{fig:linear_layout_bw_2}
    \end{figure}
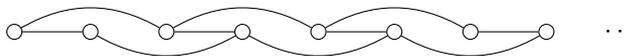
    Let $G$ be a $2\times \ell$ grid.
    \Cref{fig:linear_layout_bw_2} shows that $2\times \ell$ grid has bandwidth $2$. For treedepth, $G$ clearly doesn't contain a path on more than $|V(G)|=2^{\log(2\ell)}$ vertices, hence by the result of Hatzel et al.~\cite[Theorem 3]{Hatzel2024}, $G$ has treedepth at most $20\operatorname{log}(2\ell)=O(\log \ell)$. To see the upper bound on vertex integrity, consider partitioning the $2\times \ell$ grid into $2 \times \lfloor\sqrt{\ell}\rfloor$ grids. Consider the set $S$ consisting of two vertices connected by an edge that are first inside each block. Such set $S$ is of size at most $O(\sqrt{\ell})$ and each component of $G\setminus S$ is of size at most $2\sqrt{\ell}$, hence $\vi(G)\leq O(\sqrt{\ell})$
\end{proof}
}%

Since the resulting graph of \Cref{construction:sat} is a subgraph of a ladder of size $O(n'+m')$, it is straightforward to obtain the following corollary:
\begin{corollary}\label{cor:hardness_op_bi_mdeg3_bw2_fv1}
    \SDpath is \NP-hard even when restricted to outerplanar bipartite graphs of maximum degree $3$ and bandwidth $2$ even with $\mu \leq 1$.
\end{corollary}

\begin{apprestatable}{corollary}{corollaryHardnessUintLadderBlock}\label{cor:hardness_uint_ladder_block}
    \SDpath remains \NP-hard even when restricted to the classes of unit interval graphs, ladder graphs, or block graphs even when $\mu \leq 1$.
\end{apprestatable}
\toappendix{
\sv{\corollaryHardnessUintLadderBlock*}
\begin{proof}
    For unit interval graphs, observe that the graph~$G$ resulting from \Cref{construction:sat} is a subgraph of some unit interval graph~$G'$.
Let $G''$ be the graph obtained from~$G'$ by attaching a sufficiently long path to $s$ and let $s'$ denote the other endpoint.
The vertices on the $s'$-$s$-path simply delete one by one edges from~$E(G')\setminus E(G)$.
Any $f$-conforming $s'$-$t$-path must first traverse the $s'$-$s$-path and upon arrival at~$s$, only edges from~$G$ are available.

The idea for ladder graphs is similar.
The graph~$G$ resulting from \Cref{construction:sat} is a subgraph of a ladder graph~$G'$.
Add sufficiently long part before~$s$ that takes care of pruning~$G'$ down to~$G$.

For block graphs, use the same idea. Fill in all missing edges in the blocks of the graph $G$ resulting from \Cref{construction:sat}. Call this graph $G'$. Create graph $G''$ from $G'$ by appending a sufficiently long $s'$-$s$ path before $s$, where each vertex deletes the extra edges added to~$G$.
\end{proof}
}

We can also easily obtain hardness on cliques by adding all non-existent incident edges to the deletion set of every vertex. However, as we will see later, we cannot hope to upper bound~$\mu$ by a constant in cliques and preserve \NP-hardness (\Cref{obs:fpt_on_cliques_worse,cor:fpt_on_cliques_singleexp}).

\begin{apprestatable}{corollary}{corollaryHardnessCliques}\label{cor:hardness_cliques}
    \SDpath is \NP-hard even when restricted to cliques.
\end{apprestatable}
\toappendix{
\sv{\corollaryHardnessCliques*}
\begin{proof}
    Reduce from \SDpath, which we know is \NP-hard by \Cref{cor:hardness_fv1}. Given instance $(G,f,s,t)$ of \SDpath, we create an equivalent instance $(G',f',s,t)$ where $G'$ is a clique. We let $V(G')=V(G)$ and $E(G')=\binom{V(G')}{2}$. We let $f'(v)=f(v)\cup \{\{v,w\}\mid \{v,w\}\notin E(G)\}$ for every $v\in V(G)$. In other words, every vertex deletes incident edges that were nonexistent in the original graph $G$. Clearly there is an $f$-conforming $s$-$t$ path in $(G,f)$ if and only if there is an $f'$-conforming $s$-$t$ path in $(G',f')$.
\end{proof}
}

\subparagraph{Cactus graphs.}
\Cref{cor:hardness_op_bi_mdeg3_bw2_fv1,cor:hardness_uint_ladder_block} show that even if we allow a slight generalization of trees, the problem becomes \NP-hard. One particular non-trivial case that allows for a polynomial-time algorithm is the class of cactus graphs. We show that when restricted to cactus graphs, \SDpath becomes polynomial-time solvable (\Cref{thm:cactus_algorithm}). Interestingly, deciding an existence of an $f$-conforming path is easy, however \SSDpath remains \NP-hard even in cactus graphs (\Cref{thm:hardness_cactus_shortest}). 
Note that in cactus graphs each edge lies on at most $1$ cycle, whereas the graph resulting from \Cref{construction:sat} has each edge on at most $2$ cycles so the problem becomes \NP-hard in this case.

\begin{apprestatable}{theorem}{theoremCactusAlgorithm}\label{thm:cactus_algorithm}
    \SDpath can be solved in linear time if the underlying graph is a cactus.
\end{apprestatable}
\sv{
\begin{proof}[Proof Sketch]
    Consider the block-cut tree of $G$ and note that we can only restrict ourselves to the part of $G$ that contains the $s$-$t$ path in the block-cut tree of $G$. This part consists of cycles and bridges. In each cycle we have two options to choose from. Each choice possibly forbids some choices in the future. This can be encoded as an implication of the type \emph{if we choose this path, then we cannot choose some other path in the future}. We can thus reduce the problem to \twsat. Details can be found in the appendix.
\end{proof}
}

\toappendix{
\sv{\theoremCactusAlgorithm*}
\begin{proof}
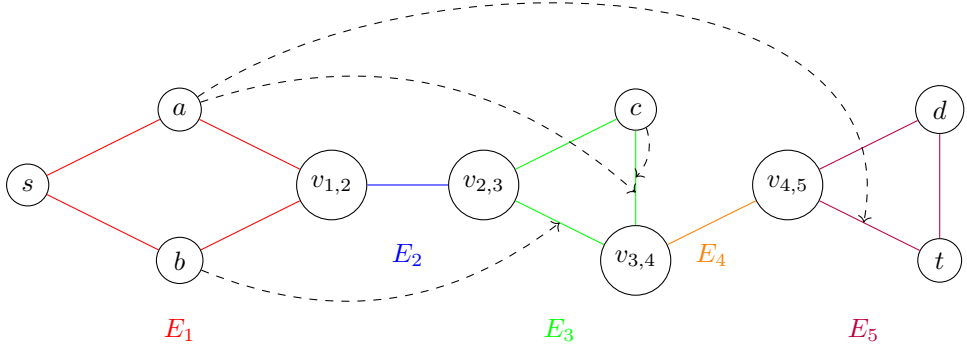
\begin{figure}[ht!]
    \centering
\begin{tikzpicture}[scale=1, every node/.style={circle, draw, minimum size=5mm}]
    \node (A) at (0,0) {$s$};
    \node (B) at (2,1) {$a$};
    \node (C) at (4,0) {$v_{1,2}$};
    \node (D) at (2,-1) {$b$};
    \node (E) at (6,0) {$v_{2,3}$};
    \node (F) at (8,1) {$c$};
    \node (G) at (8,-1) {$v_{3,4}$};
    \node (H) at (10,0) {$v_{4,5}$};
    \node (I) at (12,1) {$d$};
    \node (J) at (12,-1) {$t$}; 

    \draw[red] (A) -- (B) -- (C) -- (D) -- (A);
    \node[draw=none,below,red] at (2,-1.5) {$E_1$}; %

    \draw[blue] (C) -- (E);
    \node[draw=none,below,blue] at (5,-0.5) {$E_2$}; %

    \draw[green] (E) -- (F) -- (G) -- (E);
    \node[draw=none,below,green] at (7,-1.5) {$E_3$}; %

    \draw[orange] (G) -- (H);
    \node[draw=none,below,orange] at (9,-0.5) {$E_4$}; %

    \draw[purple] (H) -- (I) -- (J) -- (H);
    \node[draw=none,below,purple] at (11,-1.5) {$E_5$}; %

    \draw[->,dashed,bend left,out=45,in=90] (B) to (11,-0.5);
    \draw[->,dashed,bend left] (B) to (8,-0.1);
    \draw[->,dashed,bend right] (D) to (7,-0.5);
    \draw[->,dashed,bend left] (F) to (8,0.1);

\end{tikzpicture}
    \caption{Situation in the proof of \Cref{thm:cactus_algorithm}. The dashed arrows represent the deletion sets~$f(v)$. We have $P_1^1=(s,b,v_{1,2}),P_1^2=(s,a,v_{1,2}),P_3^1=(v_{2,3},v_{3,4}),P_3^2=(v_{2,3},c,v_{3,4}),P_5^1=(v_{4,5},t),P_5^2=(v_{4,5},d,t)$. The resulting \textsc{$2$-Sat} formula is: $\varphi = (x_1\Rightarrow \neg x_3)\wedge(\neg x_1\Rightarrow x_3)\wedge (\neg x_1\Rightarrow \neg x_5)\wedge (\neg x_3 \Rightarrow x_3)$. Note that the last clause has a special feature, since the $e=\{c,v_{3,4}\}$ lies on the same path as the vertex $c$ that deletes it. This immediately forbids the use of $P_3^2$. This corresponds to setting $x_3$ to \texttt{true} as the only option to satisfy $\neg x_3 \Rightarrow x_3$. A possible (and in fact unique) satisfying assignment for $\varphi$ is $x_1\mapsto \texttt{false},x_3\mapsto \texttt{true},x_5\mapsto \texttt{false}$ which corresponds to the (unique) $f$-conforming $s$-$t$ path $(s,a,v_{1,2},v_{2,3},v_{3,4},v_{4,5},d,t)$. }
    \label{fig:cactus_algorithm}
\end{figure}
    Let $(G,f)$ be a self-deleting graph, $G$ a cactus, and $s,t\in V(G)$ two vertices. We design a linear-time algorithm that either finds an $f$-conforming $s$-$t$ path or reports that none exists. Since $G$ is a cactus, each block is either a cycle or a bridge. If we contract the cycles into single vertices, we obtain a tree. 
    If there is an $f$-conforming $s$-$t$ path, then the only candidate is the unique path from $s$ to $t$ in the tree. 
    Going back to~$G$, the $s$-$t$ path is uniquely determined up to choosing which part of the cycles we traverse. 
    Let us focus on this part of $G$ consisting of sequence of bridges and cycles connecting $s$ and $t$. See \Cref{fig:cactus_algorithm} for a visualization. 
    Let $B_1, \ldots, B_k$ be the blocks in the sequence with $s \in V(B_1)$ and $t \in V(B_k)$.
    Let $V_i =V(B_i)$ 
    and denote $v_{i,i+1}$ the unique cut vertex in $V_i\cap V_{i+1}$ and further let $s=v_{0,1}$ and $t=v_{k,k+1}$. We regard blocks~$B_i$ that are cycles as a union of two vertex-disjoint paths $P_i^1$ and $P_i^2$ from $v_{i-1,i}$ to $v_{i,i+1}$.
    
    Let us also define a partial order $\leq$ on the vertices induced by the blocks as follows. Let $u\in V_i,v\in V_j$ be two vertices. If $i<j$, then $u\leq v$ and if $i=j$, then $u\leq v$ if and only if $u$ is a (not necessarily direct) predecessor of $v$ on one of the paths $P_i^1$ or $P_i^2$. For an edge $e=\{u,v\}$ and vertex $w$, by $e\leq w$ we mean $u\leq w \wedge v\leq w$.

    \begin{claim}\label{obs:cactus_alg_1}
        We can ignore edges in deletion sets $e\in f(v)$ such that $e\ngeq v$.
    \end{claim}
    \begin{claimproof}
        Observe that $s$-$t$ paths are in one-to-one correspondence with chains in the partial order $\leq$. Hence if both $e$ and $v$ appear in some $s$-$t$ path, then $e$ will precede $v$.
    \end{claimproof}
    From now on, we shall assume that for every $e\in f(v)$ we have $e\geq v$.
    \begin{claim}\label{obs:cactus_alg_2}
        We can immediately delete edges in $f(v_{i,i+1})$ for $i\in[k]\cup \{0\}$.
    \end{claim}
    \begin{claimproof}
        Any $s$-$t$ path will pass through the cut vertices $v_{i,i+1}$, hence any edges $e\geq v_{i,i+1}$ will be inevitably deleted and hence no $f$-conforming $s$-$t$ path can use these edges. Recall that we already assume that $e\geq v$ for every $e\in f(v)$.
    \end{claimproof}

    \begin{claim}\label{obs:cactus_alg_3}
        If there is a vertex $v$ and a bridge $e\in E_i$ such that $e\in f(v)$ (and $e\geq f(v)$), then we can safely remove vertex $v$.
    \end{claim}
    \begin{claimproof}
        No $f$-conforming $s$-$t$ path can use the vertex $v$ as it would otherwise delete the bridge~$e$ and there would be no $f$-conforming $v$-$t$ path in the graph, hence no $s$-$t$ path passing through $v$.
    \end{claimproof}

    We can now assume that $f(v_{i,i+1})=\emptyset$. It remains to resolve the deletion sets of the internal vertices of the paths $P_i^1$ and $P_i^2$. By \Cref{obs:cactus_alg_3} every such vertex deletes only edges~$e$ that lie on some cycle and $e\geq v$.

    We now transform \SDpath to \textsc{$2$-Sat}. We introduce a variable $x_i$ for every cycle $B_i$. The semantics of the variable is as follows. If $x_i$ is set to true, then the path should use $P_i^1$, otherwise the path should use $P_i^2$ on the cycle $B_i$. We create clauses as follows. 
    For each $i\in[k]$ and each internal vertex $v$ of  $P_i^1$ and $j\geq i$ we add the clause $x_i\Rightarrow \neg x_j$ if $f(v)\cap E(P_j^1)\neq \emptyset$ and $x_i\Rightarrow x_j$ if $f(v)\cap E(P_j^2)\neq \emptyset$. Symmetrically for internal vertices $v$ of $V(P_i^2)$ we add clauses $\neg x_i\Rightarrow \neg x_j$ if $f(v)\cap E(P_j^1)\neq \emptyset$ and $\neg x_i \Rightarrow x_j$ if $f(v)\cap E(P_j^2)\neq\emptyset$.

    It remains to say that the entire preprocessing and construction of the \textsc{$2$-Sat} instance takes $O(|V|+|E|+|f|)$ time. The resulting \textsc{$2$-Sat} instance has at most $|E|$ variables and $\sum_{v\in V}|f(v)|$ clauses and can be thus solved in time $O(|E|+|f|)$ by standard algorithms~\cite{AspvallPlassTarjan1979}.
\end{proof}
}%

\begin{apprestatable}{theorem}{theoremHardnessCactusShortest}\label{thm:hardness_cactus_shortest}
    \SSDpath is \NP-hard even when restricted to cactus graphs, and $\mu\leq 1$.
\end{apprestatable}
\sv{
\begin{proof}[Proof Sketch.]
We provide a polynomial reduction from \textsc{Independent Set}. For each vertex we introduce one cycle on the designated $s$-$t$-path (in the block-cut tree). One of the paths corresponds to taking the vertex and the other to not taking it. Inclusion of a vertex forbids inclusion of any of its neighbors by deleting an edge of the appropriate path. Crucially, the inclusion path is shorter, so the desired length of the path forces an appropriate number of vertices to be included in the independent set.
\end{proof}
}

\toappendix
{
\sv{\theoremHardnessCactusShortest*}
\begin{proof}
    We provide a polynomial reduction from \textsc{Independent Set in Cubic Graphs}, which is NP-complete~\cite{GareyJS76}.
    Here, we are given a cubic ($3$-regular) graph $G$ and an integer $k$ and the question is whether $G$ contains an independent set of size (at least)~$k$.
    We construct a cactus graph $G'$ as follows.
    For each vertex $v$ of $G$ we introduce a gadget consisting of two vertices $s^v$ and $t^v$ and two $s^v$-$t^v$-paths $P^v_0$ and $P^v_1=(s^v,a_1^v,a_2^v,a_3^v,t^v)$ of length $5$ and~$4$, respectively.
    We arbitrarily order the vertices of $G$ as $V(G)=\{v_1, v_2, \ldots, v_n\}$, for each $i \in \{1, \ldots, n-1\}$ we identify $t^{v_i}$ with $s^{v_{i+1}}$, and let $s=s^{v_1}$ and $t=t^{v_n}$.
    This finishes the construction of $G'$.
    Note that it is a cactus graph.
    For each vertex $v \in V(G)$, suppose that $N(v)=\{u_1, u_2, u_3\}$.
    Then we let $f(a_j^v)=\{\{s^{u_j},a^{u_j}_1\}\}$ for each $j \in \{1,2,3\}$.
    We leave $f(v)$ empty for all other vertices of $G'$.
    
    We claim that $G$ has an independent set of size $k$ if and only if $G'$ has an $f$-conforming $s$-$t$-path of length at most $5n-k$.
    
    We start with the only if part. 
    Let $S \subseteq V(G)$ be an independent set in $G$ of size $k$.
    We construct an $s$-$t$-path $P$ by taking for each $i \in \{1, \ldots, n\}$ the path $P_0^{v_i}$ if $v_i \notin S$ and the path $P_1^{v_i}$ if $v_i \in S$.
    As we have taken $(n-k)$ times a path of length $5$ and $k$ times a path of length~$4$, the total length of path $P$ is $5n-k$.
    We claim that $P$ is $f$-conforming.
    Indeed, suppose that there is $a_j^v \in V(P)$ such that $f(a_j^v)=\{\{s^{u_j},a_1^{u_j}\}\}$ and the edge $\{s^{u_j},a_1^{u_j}\}$ also belongs to path $P$.
    This implies that $v \in S$ and also $u_j \in S$, contradicting that $S$ is an independent set, as $u_j \in N(v)$.
    
    Now we turn to the if part.
    Suppose that $P$ is an $f$-conforming $s$-$t$-path of length at most $5n-k$ in $G'$.
    For each $i \in \{1, \ldots, n\}$, $P$ has to use either path $P^{v_i}_0$ or $P^{v_i}_1$ within the gadget of vertex $v_i$ to reach from $s^{v_i}$ to $t^{v_i}$.
    Namely, as $P^{v_i}_0$ has length $5$, $P^{v_i}_1$ has length~$4$, and $P$ has length at most $5n-k$, $P$ has to use $P^{v_i}_1$ for at least $k$ indices $i$.
    We let $S$ be the set of vertices $v_i \in V(G)$ such that $P$ uses $P^{v_i}_1$.
    By the previous argument, $S$ is of size at least $k$.
    We claim that $S$ is independent in $G$.
    Suppose it is not, namely there are $v_i$ and $v_{i'}$ with $i < i'$ such that $v_i \in S$, $v_{i'} \in S$ and $\{v_i,v_{i'}\} \in E(G)$.
    Then $v_{i'} \in N(v_i)$ and there is $j \in \{1,2,3\}$ such that $f(a_j^{v_i})=\{\{s^{v_{i'}},a_1^{v_{i'}}\}\}$, contradicting that $P$ is $f$-conforming.
    Therefore, $S$ is indeed independent.
    
    As the reduction can be clearly carried out in polynomial time, this finishes the proof.
\end{proof}
}%

\appsection{Parameterized complexity}{sec:para_complexity}
\label{sec:parameterized}

In this section we focus on parameterized complexity of (\textsc{Shortest}) \SDpath. Note that the problem is \paraNP-hard parameterized by bandwidth or maximum degree (hence also parameterized by treewidth) due to \Cref{cor:hardness_op_bi_mdeg3_bw2_fv1}. Moreover, it is also \paraNP-hard parameterized by any parameter that is constant on cliques (\Cref{cor:hardness_cliques}).

We show that \SSDpath is \W{1}-complete parameterized by $k$ (the number of vertices of the sought path) (\Cref{thm:w1_completeness_ssdpath}) and that \SDpath is \W{1}-hard parameterized by the vertex cover number. We then show that \SDpath is in fact \W{1}-complete for parameters vertex cover number, vertex integrity, treedepth, distance to linear forest, and feedback vertex set number (\Cref{thm:w1completeness_vc_dtlf_fvsn,thm:w1completeness_vi_td}).

The last hope for positive algorithmic results lies in the parameter feedback edge number ($\fen$). Here, we observe that \SDpath parameterized by $\fen$ can be solved in $O(2^{\fen}(n+m+|f|))$ time (\Cref{cor:fpt_by_fen}). Later, in \Cref{sec:kernels} we show that \SDpath even admits kernel with $O(\fen)$ vertices and edges (however, this does not yield as fast algorithm). Refer to \Cref{fig:hierarchy_structural} for a graphical overview of the results for structural parameters alone.

Then, %
in order to obtain algorithmic results, we combine structural parameters with the parameter $\mu=\max_{v\in V}|f(v)|$. 
While \SSDpath is \paraNP-hard parameterized by $\mu$ alone (\Cref{cor:hardness_fv1}), and \W{1}-complete parameterized by $k$ (\Cref{thm:w1_completeness_ssdpath}), the problem becomes \FPT parameterized by $k$ and $\mu$ combined (\Cref{thm:fpt_by_k_and_mu_better}). This also yields many \FPT results for \SDpath parameterized by structural parameters and $\mu$ combined (\Cref{thm:fpt_by_dense_and_mu}). 
Refer to \Cref{fig:hierarchy_structural_and_mu} for an overview of such results. %

\subsection{\W{1}-completeness for length of the path}

We begin with a reduction from \textsc{Multicolored Clique} to \SDpath. We developed this construction independently, although later we discovered that it is similar to the one of Bodlaender et al.~\cite[Theorem 9]{BodlaenderJK13} for a related problem.
\begin{construction}\label{construction:mcc}
    Let $G=(V,E)$ be the input graph for \textsc{Multicolored Clique} and let $V=V_1\cup \cdots \cup V_k$ be the partition of $V$ into $k$ color classes. We build an instance $(G',f,s,t)$ of \SDpath as follows. Create $k+1$ guard vertices $g_0,g_1,\ldots,g_k$. For every $i\in[k]$ and $v\in V_i$ add a vertex $y_v$ and connect it to $g_{i-1}$ and $g_i$ by edges $e_1^v=\{g_{i-1},y_v\}$ and $e_2^v=\{y_v,g_i\}$. We denote by $P_v^i$ the path $(g_{i-1},e_1^v,y_v,e_2^v,g_i)$. This completes the description of the graph $G'$. We let $s=g_0$ and $t=g_k$. It remains to specify the deletion sets. The only vertices with nonempty deletion sets will be the vertices $y_v$ on the paths $P_v^i$. We let $f(y_v)=\bigcup \{e_1^w,e_2^w \mid w\in V_j,j\neq i, \{w,v\}\notin E(G)\}$.
\end{construction}

\begin{apprestatable}{lemma}{lemmaCorrectnessMCC}\label{lem:correctness_mcc}
    Let $(G, V_1, \ldots, V_k)$ be an instance of \textsc{Multicolored Clique} and $(G',f,s,t)$ be the instance of \SDpath obtained from it by \Cref{construction:mcc}. There is a multicolored clique in $G$ if and only if there is an $f$-conforming $s$-$t$ path in $G'$.
\end{apprestatable}

\toappendix
{
\sv{\lemmaCorrectnessMCC*}
\begin{proof}
    $\Rightarrow$: If $v_{i_1},v_{i_2},\ldots,v_{i_k}$ induces a multicolored clique in $G$, then the $f$-conforming $s$-$t$ path is created by joining the guard vertices with the paths $P_{v_{i_1}}^1,P_{v_{i_2}}^2\ldots,P_{v_{i_k}}^k$. We have $\{v_{i_a},v_{i_b}\}\in E(G)$ for any distinct $a,b\in[k]$, no edges of the paths $P_{v_{i_j}}$ are deleted by the preceding vertices, hence the resulting path is $f$-conforming.

    $\Leftarrow$: Let $P$ be a $f$-conforming $s$-$t$ path. By construction, it necessarily consists of the guard vertices joined by the paths $P_{v}$ and for each $i\in[k]$ there is exactly one subpath $P_{v}$ for $v\in V_i$.
    Let $S = \bigcup_{i\in[k]}\{v\in V_i\mid \text{the segment $P_v$ is contained in $P$}\}$.
    If $u \in S \cap V_i$ was not adjacent to $v \in V_j \cap S$ for some $j > i$, then $E(P_v) \subseteq f(u)$, contradicting $P$ being $f$-conforming.
    Hence $S$ is a multicolored clique in $G$.
\end{proof}
}%

Note that in the instance $(G',f,s,t)$ from \Cref{construction:mcc}, any $f$-conforming $s$-$t$ path has at most $2k+1$ vertices. We immediately obtain \W{1}-hardness of \SSDpath
w.r.t. $k$.

\begin{apprestatable}{theorem}{theoremWOneCompletenessSSDpath}\label{thm:w1_completeness_ssdpath}
    \SSDpath is \W{1}-complete w.r.t.\ $k$.
\end{apprestatable}
\toappendix
{
\sv{\theoremWOneCompletenessSSDpath*}
\begin{proof}
    The hardness part follows from \Cref{construction:mcc} and \Cref{lem:correctness_mcc} since any $f$-conforming $s$-$t$ path in $(G',f)$ has at most $2k+1$ vertices. It remains to prove \W{1} membership. We provide a parameterized reduction from \SSDpath to \textsc{Multicolored Clique}. Let $(G=(V,E),f,s,t,k)$ be an instance of \SSDpath. We create an instance of \textsc{Multicolored Clique}  $(G',k')$ as follows. First, enhance the set $E$ with loops on each vertex, i.e., $E^*=E\cup \{e_{vv}\mid v\in V\}$, where $e_{vv}$ is the loop on vertex $v$ and denote $G^*=(V,E^*)$. The graph $G'$ is built as follows:

    \begin{enumerate}[(a)]
        \item Create $k$ copies of the vertex set $V$, denoted $V_1,V_2,\ldots,V_{k}$ and $k-1$ vertex sets $W_1,W_2,\ldots, W_{k-1}$ where each set corresponds to the set of edges $E^*$. We refer to sets $V_i$ as \emph{vertex layers} and to $W_i$ as \emph{edge layers}. 
        For vertex $v\in V$ denote $x_v^i\in V_i$ the vertex corresponding to $v$ in the $i$-th vertex layer and for $e\in E^*$ denote the corresponding vertex in the edge layer $W_i$ by $y_e^{i}$.
        \label{cons1}
        \item Connect the vertex layers into a complete $k$-partite graph: add an edge between each $x_u^i,x_v^j,i\neq j,u,v\in V$.
        \label{cons2}
        \item Connect the edge layers into a complete $(k-1)$-partite graph: add edge between each $y_e^i,y_f^j,i\neq j,e,f\in E^*$.
        \label{cons3}
        \item For each edge layer $W_i$ and $j>i+1$ connect $y_e^i$ to $x_v^j$ and $y_f^j$ for any $v\in V,e,f\in E^*$.
        \label{cons4}
        \item For each vertex layer $V_i$ and  $j>i+1$ connect $x_v^i$ to $x_u^j$ and $y_e^j$ for any $u,v\in V,e\in E^*$
        \label{cons5}
        \item Now define the edges between consecutive layers. Locally, $V_i\cup W_i$ or $V_i\cup W_{i-1}$ resemble the incidence graph of vertices and edges of $G$. 
        I.e., connect $x_v^i\in V_i$ to $y_e^i\in W_i$ and to $y_e^{i-1}\in W_{i-1}$ if $v\in e$.
        \label{cons6}
        \item Now incorporate the deletions. For any $x_v^i\in V_i$ and $j\geq i$ delete the edge from $x_v^i$ to $y_e^j$ if and only if $e\in f(v)$.
        \label{cons7}
        \item Remove all vertices except $x_s^1$ from $V_1$ and all vertices except $x_t^k$ from $V_k$.
        \label{cons8}
    \end{enumerate}
    \newcommand{\stepref}[1]{\itemstyle{\hyperref[#1]{(\ref{#1})}}}
    This finishes the construction of the instance $(G',k',V_1, \ldots, V_k, W_1, \ldots, W_{k-1})$ of \textsc{Multicolored Clique}. We have $k'=2k-1$ and $V(G')=V_1\cup V_2\cup \cdots V_k \cup W_1\cup \cdots W_{k-1}$ is the partition of vertices.

    \begin{claim}
        If there is an $f$-conforming $s$-$t$ path on at most $k$ vertices in $G$, then there is a multicolored clique in $G'$.
    \end{claim}
    \begin{proof}
        Let $P=(s=v_1,e_1,v_2,\ldots,e_{\ell-1},t=v_\ell)$ be an $f$-conforming $s$-$t$ path in $G$. We extend $P$ into an $f$-conforming $s$-$t$ walk $P^*=(v_1^*,e_1^*,v_2^*,\ldots,e_{k-1}^*,v_k^*)$ in $G^*$ on exactly $k$ vertices by adding loops on vertex $t$ if necessary. Let $X^V=\{x_{v_i^*}^i\mid i\in[k]\},Y^W=\{y_{e_i^*}^i\mid i\in[k-1]\}$. We show that $G'[X^V\cup Y^W]$ is a clique. To see this, note that $G[X^V]$ and $G[Y^W]$ are cliques by construction steps \stepref{cons2},\stepref{cons3}. By construction steps \stepref{cons4},\stepref{cons5} every $x_{v_i^*}$ is connected to any $y_{e_j^*}$ if $|j-i|>1$. Since $P^*$ is a path, there are also $G'$-edges for $|j-i|\leq 1$ by construction step \stepref{cons6}. Finally, since $P^*$ is $f$-conforming, construction step \stepref{cons7} deletes none of the edges.
    \end{proof}

    \begin{claim}
        If there is a multicolored clique in $G'$, then there is an $f$-conforming $s$-$t$ path on at most $k$ vertices in $G$.
    \end{claim}
    \begin{proof}
        Let $x_{v_i}^i \in V_i$ and $y_{e_i}^i\in W_i$ be the vertices of the multicolored clique. We claim that the walk $P=(v_1,e_1,v_2,\ldots,e_{k-1},v_k)$ is $f$-conforming. As it has exactly $k$ vertices, it can be shortened to an $f$-conforming path on at most $k$ vertices. Since the only vertex in $V_1$ is $x_s^1$ and the only vertex in $V_k$ is $x_t^1$ (construction step \stepref{cons1}, we have $v_1=s$ and $v_k=t$. Incidence of edges $e_i$ with $v_{i-1}$ and $v_i$ are due to construction step \stepref{cons6}. Note that if $e_i\in f(v_j)$ for some $j\leq i$, then there is no edge from $x_v^i$ to $y_e^j$ by construction step \stepref{cons7}, contradicting that the vertices $x_{v_i}^i$ and $y_{e_i}^i$ induce a clique.
    \end{proof}
As the construction can be carried out in polynomial time, this finishes the proof.
\end{proof}
}%
\subsection{Parameterization by structural parameters alone}
Observe that the resulting graph $G'$ from \Cref{construction:mcc} has vertex cover number at most $k+1$, because $G'\setminus \{g_0,\ldots,g_k\}$ is edgeless. We immediately obtain \W{1}-hardness of \SDpath for the parameter vertex cover number. 

We establish membership in \W{1} of \SDpath parameterized by the feedback vertex set number  (\Cref{thm:membership_w1_fvsn}) and treedepth (\Cref{obs:w1_membership_td}). These results together with \W{1}-hardness for vertex cover number establish \W{1}-completeness for the following parameters: vertex cover, distance to linear forest, feedback vertex set, vertex integrity, and treedepth (\Cref{thm:w1completeness_vc_dtlf_fvsn,thm:w1completeness_vi_td}).

\begin{apprestatable}{theorem}{theoremMembershipWOneFVSN}\label{thm:membership_w1_fvsn}
    \SDpath parameterized by the feedback vertex set number is in \W{1}.
\end{apprestatable}
\sv{
\begin{proof}[Proof Sketch]
    If $G$ has a feedback vertex set $S$, then any path in $G$ is split by $S$ into at most $|S|+1$ segments where the segments are uniquely determined by their endpoints, as $G\setminus S$ is a forest. We can thus equivalently look for a path of length $O(\fvsn)$ in a graph where we represent long paths in $G\setminus S$ by paths of length two and reflect the deletions along these unique $u$-$v$ paths in $G\setminus S$. Full proof can be found in the appendix.
\end{proof}
}
\toappendix{
\sv{\theoremMembershipWOneFVSN*}
\begin{proof}
    We reduce \SDpath parameterized by \fvsn to \SSDpath parameterized by $k$ which is in \W{1} by \Cref{thm:w1_completeness_ssdpath}. Let $(G,f,s,t)$ be the input instance to \SDpath parameterized by \fvsn, we build a new instance $(G',f',s',t',k')$ as follows. 
    
    Let $S\subseteq V(G)$ be the modulator to forest of size \fvsn and denote $F=G\setminus S$ (note that $S$ can be computed in \FPT time parameterized by its size~\cite{KociumakaP2014}). For $u,v\in V(F)$ let $P_{u,v}$ denote the unique path from $u$ to $v$ in $F$ (if it exists).

    We let $s'=s,t'=t$. We build a $1$-subdivided clique on the vertices of $F$. More precisely, for any $\{u,v\}\subseteq V(F)$ we create a path $Q_{u,v}=(u,e_1^{uv},y_{uv},e_2^{uv},v)$, where $y_{uv}$ is a new vertex. All edges with endpoints in $S$ remain untouched. Formally, we have $V(G')=V(G)\cup \{y_{uv}\mid \{u,v\}\subseteq V(F)\}$ and $E(G')=\{\{x,y\}\in E, \{x,y\}\cap S\neq \emptyset\} \cup \{e_1^{uv},e_2^{uv}\mid \{u,v\}\subseteq V(F)\}$.
    
    Now we deal with the deletion sets.
    \begin{enumerate}[(a)]
        \item Each vertex retains deletions of edges incident to $S$.\label{itm:fvsn_step1}
        \item For each vertex $v\in V(G)$ of the original graph and $\{x,y\}\subseteq V(F)$, if $E(P_{x,y})\cap f(v)\neq \emptyset$, then add the edges $e_1^{xy},e_2^{xy}$ to $f'(v)$.\label{itm:fvsn_step2}
        \item The middle vertices $y_{uv}$ delete everything that is deleted by the inner vertices of the path $P_{u,v}$. Suppose that for some $z\in V(P_{u,v})\setminus \{u,v\}$ we have $e\in f(z)$. If $e \in E(F)$, then for every $\{x,y\}\subseteq V(F)$, $\{x,y\} \neq \{u,v\}$ such that $e\in E(P_{x,y})$, we add $e_{1}^{xy}$ and $e_2^{xy}$ to $f'(y_{uv})$. If $e$ was incident to $S$, then simply add $e$ to $f'(y_{uv})$.\label{itm:fvsn_step3}
        \item Lastly, for any $\{u,v\}\subseteq V(F)$, if the path $P_{u,v}$ is not $f$-conforming, then $y_{uv}$ also deletes the edge $e_1^{uv}$ and if $P_{v,u}$ is not $f$-conforming, then $y_{uv}$ also deletes the edge $e_2^{uv}$. \label{itm:fvsn_step4}
    \end{enumerate}
    \newcommand{\stepref}[1]{\itemstyle{\hyperref[#1]{(\ref{#1})}}}
    Note that in particular, if the path $P_{u,v}$ does not exist in $F$, then in particular it is not $f$-conforming, hence $y_{uv}$ disallows passing in both directions by \stepref{itm:fvsn_step4}.
    Finally, set $k'=4\fvsn + 3$.

    \begin{claim}
        If there exists an $f$-conforming $s$-$t$ path in $G$, then there is an $f$-conforming $s'$-$t'$ path in $G'$ on at most $k'$ vertices.
    \end{claim}
    \begin{proof}
        Let $P=(s=v_1,e_1,v_2,\ldots,e_{\ell-1},v_{\ell}=t)$ be an $f$-conforming $s$-$t$ path in $G$. The modulator $S$ splits $P$ into subpaths $P_1,P_2,\ldots,P_q$ for some $q\leq |S|+1$,
        where $P_i=(v^i_1,e_1^i,v_2^i\ldots,e_{\ell_i-1}^i,v_{\ell_i}^i)$ and the $i$-th and $(i+1)$-th subpath are joined via edges $\{v^i_{\ell_i},x\},\{x,v^{i+1}_1\}$ for some $x\in S$. The desired path in $G'$ is created by replacing every subpath $P_i$ with $|V(P_i)|\geq 2$ by the path
        $Q_{v_1^i,v_{\ell_i}^i}$. Note that the edges of the path $Q_{v_1^i,v_{\ell_i}^i}$ are not deleted by $v_1^i$ nor by any vertices preceeding $v_1^i$. To see this, note that if some vertex $v$ deleted an edge $e_1^{v_1^iv_{\ell_i}^i}$ or $e_2^{v_1^iv_{\ell_i}^i}$ this means that $E(P_{x,y})\cap f(v)=\emptyset$ by \stepref{itm:fvsn_step2}, contradicting that $P$ was originally $f$-conforming. Next, 
        $y_{v_1^i,v_{\ell_i}^i}$ cannot delete $e^{v_1^i,v_{\ell_i}^i}_2$ as the subpath $P_i$ was also $f$-conforming, so \stepref{itm:fvsn_step4} did not apply. Vertex $y_{v_1^i,v_{\ell_i}^i}$ also could not delete any future edges after $v_{\ell_i}$ as otherwise, by \stepref{itm:fvsn_step3} the inner vertices of the path $P_i$ deleted some $e\in E(P_{x,y})$ in some future subpath $P_{x,y}$, or it was the case of an edge incident to $S$, but that is also in contradiciton with the entire path $P$ being $f$-conforming (due to \stepref{itm:fvsn_step1}).
        Hence the resulting path is $f$-conforming. Each segment now contains at most $3$ vertices and there are at most $\fvsn+1$ segments. Altogether the resulting path in $G'$ has at most $3(\fvsn + 1)+\fvsn = 4\fvsn + 3=k'$ vertices.
    \end{proof}
    \begin{claim}
        If there exists an $f$-conforming $s'$-$t'$ path in $G'$ on at most $k'$ vertices, then there exists an $f$-conforming $s$-$t$ path in $G$.
    \end{claim}
    \begin{proof}
        Let $P'=(s=v_1',e_1',v_2',\ldots,e'_{r-1},v'_r=t)$ be an $f'$-conforming $s$-$t$ path in $G'$ on at most $k'$ vertices. Whenever $v'_i=y_{uv}$ for some $\{u,v\}\subseteq V(F)$, then necessarily $(v'_{i-1},e_{i-1}',v'_i,e_i',v'_{i+1})=Q_{u,v}$. Replace this segment by the path $P_{u,v}$ in the original graph~$G$. We possibly create a walk. 

        We now argue that the resulting path (walk) after replacing all such segments is $f$-conforming. First, note that if the path $P_{u,v}$ in $G$ did not exist (for example because $u,v$ are in different connected components of $G\setminus S$), then in particular neither the path $P_{u,v}$ nor $P_{v,u}$ is $f$-conforming and in this case, due to \stepref{itm:fvsn_step4}, the middle vertex $v_i'$ deleted $e_i'$, which is not possible. Moreover, this rules out the possibility that some vertex on the path $P_{u,v}$ deletes some future edge of the path $P_{u,v}$ as otherwise \stepref{itm:fvsn_step4} applied and again $v_i'$ would delete $e_i'$.

        Next, we verify that no vertex $x$ before $v_{i-1}'$ deleted an edge of $P_{u,v}$. To see this note that this would imply, by \stepref{itm:fvsn_step3} that $x$ in $P'$ deleted the edges $e_1^{uv}$ and $e_2^{uv}$, again contradicting $f$-conformity of $P'$.

        Finally, we verify that the inner vertices of $P_{u,v}$ do not delete any edges in the future. To see this, note that in that case \stepref{itm:fvsn_step3} applied and in this case $e_1^{xy},e_2^{xy}$ were added to the deletion set of $v_i'$, or again it was the case of an edge incident to $S$, which again contradicts the assumption that $P'$ was $f$-conforming.

        Finally, we shorten the $f$-conforming $s$-$t$ walk in $G$ to an $f$-conforming $s$-$t$ path in $G$ and this finishes the proof of the claim.
    \end{proof}
    As the reduction can be performed in fpt-time, this finishes the proof.
\end{proof}
}%

\begin{apprestatable}{theorem}{theoremWOneCompletenessVcDtlfFvsn}\label{thm:w1completeness_vc_dtlf_fvsn}
    \SDpath parameterized by feedback vertex set, distance to linear forest, or vertex cover is \W{1}-complete, solvable in $n^{O(\alpha)}$ time and unless ETH fails, there is no $g(\alpha) n^{o(\alpha)}$ algorithm for any of the above parameters and any computable function $g$.
\end{apprestatable}

\toappendix{
\sv{\theoremWOneCompletenessVcDtlfFvsn*}
\begin{proof}
    \SDpath is \W{1}-hard w.r.t.\ vertex cover number because the resulting graph $G'$ from \Cref{construction:mcc} has vertex cover number at most $k+1$. This is because $G'\setminus \{g_0,\ldots,g_k\}$ is edgeless. \W{1} membership for \fvsn
    follows from \Cref{thm:membership_w1_fvsn}. Clearly $\SSDpath$ can be solved in $n^{O(k)}$ time be guessing all possible $k$-tuples of vertices representing the path. The running times of the algorithms for other parameters and the lower bounds follow from the fact that all the involved reductions are linear in the parameter.
\end{proof}
}%

\begin{apprestatable}{lemma}{lemmaLengthOfPathInTdAndVi}\label{lem:length_of_path_in_td_and_vi}
    Let $G$ be a graph with treedepth $\td$ and vertex integrity $\vi$. Then $G$ contains no path on more than $2^{\td}$ or $\vi^2+2\vi$ vertices.
\end{apprestatable}
\toappendix
{
\sv{\lemmaLengthOfPathInTdAndVi*}
\begin{proof}
    By the result of Nešetřil and Ossona de Mendez~\cite[Chapter 6.2]{SPARSITY}, an $n$-vertex path has treedepth equal to $\lceil \log_2(n+1)\rceil$. Hence if $G$ contained a path $P$ on more than $2^{\td}$ vertices, then since treedepth is monotone under taking subgraphs, it follows that $\td(G)\ge \td(P) \ge \lceil\log_2(2^{\td}+2)\rceil>\td$, a contradiction.

    For vertex integrity, we show that if $G$ has vertex integrity $\vi$, then $G$ contains no path on more than $\vi^2+2\vi$ vertices. To see this, note that if $S$ is the modulator of size $\vi(G)$ such that $G\setminus S$ has components of size at most $\vi(G)$, then any path is split by $S$ into at most $\vi(G)+1$ segments. Each such segment can have at most $\vi(G)$ vertices as it has to belong to a connected component of $G\setminus S$. Hence there are at most $\vi(G)+1$ segments of size $\vi(G)$ plus the vertices of $S$, which in total gives an upper bound of $\vi(G)+\vi(G)(\vi(G)+1)$ on the number of vertices of any path in $G$.
\end{proof}
}%

\begin{observation}\label{obs:w1_membership_td}
    \SDpath parameterized by treedepth is in \W{1}.
\end{observation}
\begin{proof}
    We provide a parameterized reduction from \SDpath parameterized by treedepth to \SDpath parameterized by $k$, which is in \W{1} by \Cref{thm:w1_completeness_ssdpath}. Let $(G,f,s,t)$ be an instance of \SDpath, we return the instance $(G,f,s,t,k)$ where $k=2^{\td(G)}$. Correctness follows from \Cref{lem:length_of_path_in_td_and_vi}.
\end{proof}

\begin{apprestatable}{theorem}{theoremCompletenessLowerboundVertexIntegrityTreedepth}\label{thm:w1completeness_vi_td}
    \SDpath is \W{1}-complete parameterized by treedepth or vertex integrity. More precisely,
    \SDpath can be solved in $n^{O(2^{\td})}$ and $n^{O(\vi^2)}$ time.
    For any $\varepsilon>0$, algorithms for \SDpath with running times 
    $n^{O(\vi^{2-\varepsilon})}\poly(n)$ or  $n^{2^{o(\td)}}\poly(n)$ violate ETH.
\end{apprestatable}
\toappendix
{
\sv{\theoremCompletenessLowerboundVertexIntegrityTreedepth*}
\begin{proof}
    For the algorithms reduce \SDpath to \SSDpath for $k=2^{\td}$ and $k = \vi^2+2\vi$, respectively. Correctness follows from \Cref{lem:length_of_path_in_td_and_vi}. For the lower bound, recall that \Cref{construction:sat} yields a subgraph of $2\times \ell$ grid where $\ell=O(n'+m')$ where $n'$ and $m'$ are the number of variables and clauses of the original \tsat formula (see also \Cref{rem:sat_is_linear}). \Cref{lem:structural_properties_grid} establishes that $2\times \ell$ grid has logarithmic treedepth and vertex integrity of $O(\sqrt{\ell})$. Therefore, the algorithms with running times 
    $n^{O(\vi^{2-\varepsilon})}\poly(n)$ or  $n^{2^{o(\td)}}\poly(n)$ would imply $2^{o(n'+m')}$-time algorithms for \tsat, violating ETH.
\end{proof}
}%

We complement the hardness results by a positive result for the parameter feedback edge number. Note that given a path $P$ in a self-deleting graph, it is easy to check whether $P$ is $f$-conforming in time $O(n+m+|f|)$.

\begin{apprestatable}{lemma}{lemmaFenTwoToKManyPaths}\label{lem:fen_two_to_k_many_paths}
    Let $G$ be a graph and $s,t\in V(G)$ two fixed vertices in $G$. Then the number of $s$-$t$-paths in $G$ is at most $2^{\fen(G)}$.
\end{apprestatable}
\toappendix{
\sv{\lemmaFenTwoToKManyPaths*}
The proof of this lemma is similar to that of Demaine et al.~\cite[Section 4]{DemaineEHJLUU19} (see also the ArXiv version, Section 4.1).

\begin{proof}
    Fix one $s$-$t$ path $P$. Consider the $\mathbb{Z}_2$-vector space $\mathcal{C}(G)$ of Eulerian subgraphs of $G$ (also known as the \emph{cycle space} of $G$). Observe that any symmetric difference of the edge set of $P$ with some other $s$-$t$ path $P'$ (different from $P$) yields an unique Eulerian subgraph of $G$ with at least one edge. There are $2^{\dim \mathcal{C}(G)}-1$ possible Eulerian subgraphs of $G$ with at least one edge. It is a well-known fact that $\dim \mathcal{C}(G) = |E(G)|-|V(G)|+\operatorname{cc}(G)=\fen(G)$. Together with $P$ we obtain the total number of $2^{\fen(G)}$ $s$-$t$ paths in $G$, as claimed.
\end{proof}
}%

\begin{apprestatable}{corollary}{corollaryFptByFen}\label{cor:fpt_by_fen}
     \SDpath can be solved in $O(2^{\fen(G)}(n+m+|f|))$ time. Moreover, $2^{o(\fen)}\poly(n)$-time algorithm for \SDpath violates ETH.
\end{apprestatable}
\toappendix{
\sv{\corollaryFptByFen*}
\begin{proof}
    The algorithm immediately follows from \Cref{lem:fen_two_to_k_many_paths}. It is not hard to observe that we can also enumerate all the paths in $2^{\fen}O(n+m+|f|)$ time and check in $O(n+m+|f|)$ time per path whether it is $f$-conforming. The lower bound comes from the fact that the instance produced by \Cref{construction:sat} has $\fen = O(n)$, hence an $2^{o(\fen)}\poly(n)$ algorithm yields $2^{o(n'+m')}$ algorithm for \tsat (see also \Cref{rem:sat_is_linear}).
\end{proof}
}

\subsection{\FPT algorithm for $k$ and $\mu$ combined}
We demonstrate that \SSDpath parameterized by $k$ and $\mu$ combined is in \FPT. We utilize color-coding, \lv{originally }introduced by Alon\sv{ et al.}\lv{, Yuster and Zwick}~\cite{AlonYZ1995}\lv{ to solve the problem of finding a simple path in a graph on (at least) $k$ vertices in $2^{O(k)}m\log n$ time}. \lv{The general idea is as follows. }We consider a colorful variant of the problem, where we consider coloring of the edges and we only distinguish edges based on their colors. The sought solution -- an $f$-conforming path $P=(v_1,e_1,\ldots,e_{r-1},v_r)$ on $r \leq k$ vertices interacts with the edges $e_i$ in the path and with the edges in the deletion sets $f(v_1),\ldots,f(v_r)$. By assumption there are at most $\mu k+k-1$ such edges in total. Hence if we can ensure that the coloring will behave nicely on the set $\{e_1,\ldots,e_{r-1}\}\cup \bigcup_{i=1}^rf(v_i)$, we will find the path even in the colorful variant. 
We now formally define the colorful variant of our problem, which we refer to as \SCHCpath.

\begin{restatable}{definition}{definitionRainbow}
Let $(G,f)$ be a self-deleting graph and let $\chi\colon E(G)\to [q]$ be a coloring of its edges. Let $P=(v_1,e_1,\ldots,e_{r-1},v_r)$ be a path in $G$. We say that 
\begin{enumerate}
    \item $P$ is \emph{$\chi$-compliant} if $\chi(e_i)\notin \chi(f(v_j))$ for any $j\leq i$,
    \item $P$ is \emph{half-$\chi$-rainbow} if $\chi(\bigcup_{i=1}^{r}f(v_i)\setminus E(P))\cap \chi (E(P))=\emptyset$ and $\chi|_{E(P)}$ is injective.
    \item $P$ is \emph{$\chi$-rainbow} if for $F=E(P)\cup \bigcup_{i=1}^rf(v_i)$ the restriction $\chi|_{F}$ is injective.
\end{enumerate}
\end{restatable}

In the \SCHCpath problem we are given a self-deleting graph $(G,f)$, positive integer $k$, coloring $\chi \colon E(G)\to [q]$, vertices $s,t\in V(G)$ and the task is to decide whether there is a $\chi$-compliant $s$-$t$ path on at most $k$ vertices in $G$.

\begin{apprestatable}{lemma}{lemmaBasicRainbowProperties}\label{obs:basic_properties_chi_compl_conf_rainbow}
    The following holds for any path $P$:
    \begin{enumerate}
        \item $P$ is $\chi$-rainbow $\Rightarrow$ $P$ is half-$\chi$-rainbow.
        \item $P$ is half-$\chi$-rainbow and $f$-conforming $\Rightarrow$ $P$ is $\chi$-compliant.
        \item $P$ is $\chi$-compliant $\Rightarrow$ $P$ is $f$-conforming.
    \end{enumerate}
\end{apprestatable}
\toappendix{
\sv{\definitionRainbow*}
\sv{\lemmaBasicRainbowProperties*}
\begin{proof}
    \begin{enumerate}
        \item Clearly, since $\chi$ is injective on $F=E(P)\cup \bigcup_{i=1}^rf(v_i)$, it is also injective on $E(P)\subseteq F$. Let $A=E(P),B=\bigcup_{i=1}^rf(v_i)$. Since $(A\setminus B), B$ is a partition of $A\cup B$, and $\chi|_{A\cup B}$ is injective, clearly $\chi(A\setminus B)\cap \chi(B)=\emptyset$.
        \item Suppose for the sake of contradiction that $P$ is not $\chi$-compliant. That is, for some $j\leq i$ we have $\chi(e_i)\in \chi(f(v_j))$, i.e., there is an edge $e^*\in f(v_j)$ such that $\chi(e_i)=\chi(e^*)$. Now, either $e^*$ is outside $P$, which contradicts the assumption that $\chi(\bigcup_{i=1}^{r}f(v_i)\setminus E(P))\cap \chi (E(P))=\emptyset$, because $e_i\in E(P)$. Or $e^*$ is on $P$, thus $e_i=e^*$ by injectivity of $\chi|_{E(P)}$, contradicting the assumption that $P$ was $f$-conforming.
        \item For the sake of contradiction, suppose that $P$ is not $f$-conforming, i.e., there are indices $j\leq i$ such that $e_i\in f(v_j)$. But this implies that $\chi(e_i)\in \chi(f(v_j))$, contradicting the assumption that $P$ is $\chi$-compliant.
    \end{enumerate}
\end{proof}

}
\begin{apprestatable}{lemma}{lemmaAlgorithmQColors}\label{lem:algorithm_Q_colors}
    Given an instance of \SCHCpath and a set of colors $Q\subseteq [q]$, we can in $O(2^{|Q|}(n+m+|f|))$ time decide whether there exists a $\chi$-compliant path on at most $k$ vertices using only colors from $Q$.
\end{apprestatable}

\toappendix{
\sv{\lemmaAlgorithmQColors*}
\begin{proof}
    We construct an auxiliary simple directed graph $G'$ representing a state space as follows. A vertex of $G'$ is a pair $(v,Q')$ where $Q'\subseteq Q,v\in V(G)$, representing the fact that we can reach $v$ from $s$ with colors from $Q'$ still available. The edges of $G'$ are as follows. There is a directed edge from $(u,Q_1)$ to $(v,Q_2)$ if $\{u,v\}\in E(G)$, $\chi(\{u,v\})\in Q_1$ and $Q_2 = Q_1\setminus \chi(f(v))$.
    
    \begin{claim}
        There is a $\chi$-compliant $s$-$t$ path on at most $k$ vertices in $G$ if and only if there is some $Q_t\subseteq Q$ such that there is a $(s,Q_s)$-$(t,Q_t)$ path on at most $k$ vertices in $G'$, where $Q_s=Q\setminus \chi(f(s))$.
    \end{claim}
    
    \begin{claimproof}
        $\Rightarrow$: Let $P=(s=v_1,e_1,v_2,\ldots,e_{r-1},v_r=t)$ for $r\leq k$ be a $\chi$-compliant $s$-$t$ path in~$G$. The sequence of vertices of the $(s,Q_s)$-$(t,Q_t)$ path in $G'$ are given by $(s,Q_s),(v_2,Q_s\setminus \chi(f(v_2))),\ldots,(v_i,Q_s\setminus \bigcup_{j\le i}\chi(f(v_j))),\ldots,(t,Q_t=Q_s\setminus \bigcup_{j \le r}\chi(f(v_j)))$. It is straightforward to verify from the definition of $G'$ that this is indeed a path in $G'$.

        $\Leftarrow$: Let $P=((s=v_1,Q_1=Q_s),(v_2,Q_2),\ldots,(v_r=t,Q_r=Q_t))$, $r\leq k$ be a path in $G'$. We show that $(v_1,v_2,\ldots,v_r)$ is a $\chi$-compliant walk, which indeed can be shortened to a $\chi$-compliant path on at most $k$ vertices. Clearly $e_i=\{v_i,v_{i+1}\}\in E(G)$ by the definition of $G'$.
        We verify that for all $j\leq i$ we have $\chi(e_i)\notin \chi(f(v_j))$. Observe that $Q_r\subseteq Q_{r-1}\subseteq \cdots \subseteq Q_1$ by the definition of $G'$. For the sake of contradiction suppose that $\chi(e_i)\in \chi(f(v_j))$.
        Since $Q_j = Q_{j-1}\setminus \chi(f(v_j))$ (or $j=1$ and $Q_j=Q\setminus \chi(f(s))$), we obtain that $\chi (e_i)\notin Q_j$. And since $Q_i\subseteq Q_j$, it follows that $\chi(e_i)\notin Q_i$. Recall that $e_i=\{v_i,v_{i+1}\}$ and since $\chi(\{v_i,v_{i+1}\})\notin Q_i$, there is no edge from $(v_i,Q_i)$ to $(v_{i+1},Q_{i+1})$, contradicting that $P$ was a path in $G'$.
    \end{claimproof}
    To find an  $(s,Q_s)$-$(t,Q_t)$ path on at most $k$ vertices in $G'$ we can use the standard BFS algorithm. It remains to bound the running time. Note that the number of vertices and edges of $G'$ is at most $2^{|Q|}n$ and $2^{|Q|}m$, respectively, and we can build it in $O(2^{|Q|}(n+m+|f|))$ time.
\end{proof}
}

\begin{lemma}\label{lem:chi_compliant_running_time}
    \SCHCpath can be solved in $O(2^q(n+m+|f|))$ or in ${O(\binom{q}{k}2^k(n+m+|f|))}$ time.
\end{lemma}
\begin{proof}
    We can either plug $Q=[q]$ into the algorithm of \Cref{lem:algorithm_Q_colors} and obtain the running time $O(2^q(n+m+|f|))$ or we can try all possible sets $Q$ of $k-1$ colors and \lv{plug it into \Cref{lem:algorithm_Q_colors} to }obtain the running time $O(\binom{q}{k}2^k(n+m+|f|))$. Note that we seek a path on (at most) $k$ vertices, hence (at most) $k-1$ edges.
\end{proof}

\begin{apprestatable}{theorem}{theoremRandomizedColorCodingAlg}\label{thm:randomized_color_coding_alg}
    There is a randomized algorithm solving \SSDpath with the following guarantees. Given $\varepsilon\in(0,1)$, it runs in $2^{O(k\log \mu)}(n+m)\ln \frac{1}{\varepsilon}$ time. Moreover, if the input is a no-instance, the algorithm outputs no. If the input is a yes-instance, the algorithm outputs yes with probability at least $1 - \varepsilon$. 
\end{apprestatable}
\sv{\begin{proof}[Proof sketch.]
    Asume that $\mu \geq 1$ as otherwise the problem is trivial. We use the algorithm from \Cref{lem:chi_compliant_running_time} with running time $O(\binom{q}{k}2^k(n+m+|f|))$. We try a random edge coloring $\chi$ using $q=4k\mu$ colors. For a fixed coloring using $q$ colors, we lower bound the probability that a path $P$ is $\chi$-compliant, given that $P$ is $f$-conforming. Thus, we are equivalently lower bounding the probability that the algorithm succeeds in finding the $f$-conforming path. Let $P=(v_1,e_1,\ldots,e_{r-1},v_r)$. We lower bound the probability that $P$ is half-$\chi$-rainbow. Suppose that the edges in $\bigcup_{i=1}^r f(v_i)$ are colored by $w\leq |\bigcup_{i=1}^r f(v_i)\setminus E(P)|\leq r\mu \leq k\mu$ colors from the set $[q]$. For $P$ to become half-$\chi$-rainbow, we need to ensure that the edges in $E(P)$ are colored by one of the remaining $q-w\geq 4k\mu-k\mu = 3k\mu$ colors and that the coloring is injective on $E(P)$. The probability of that happening for the $r-1$ edges is
    $\frac{q-w}{q}\cdot \frac{q-w-1}{q}\cdot \cdots \cdot \frac{q-w-r+2}{q}$.
    Note that the last term lower bounds every other term and the last term is lower bounded by $\frac{1}{2}$ because $\frac{q-w-r+2}{q}\geq \frac{3k\mu-k}{4k\mu}\geq \frac{1}{2}$ given that $\mu \geq 1$. Hence the probability that $P$ is half-$\chi$-rainbow is at least $2^{-r+1}\geq 2^{-k}$.
     
    We repeat the above for $2^{k}\ln \frac{1}{\varepsilon}$ random choices of $\chi$ and the proof on the probability of error and running time follow. The full proof can be found in the appendix. 
\end{proof}}

\toappendix{
\sv{\theoremRandomizedColorCodingAlg*}
\begin{proof}
    Asume that $\mu \geq 1$ as otherwise the problem is trivial. We use the algorithm from \Cref{lem:chi_compliant_running_time} with running time $O(\binom{q}{k}2^k(n+m+|f|))$. We try a random edge coloring $\chi$ using $q=4k\mu$ colors. For a fixed coloring using $q$ colors, we lower bound the probability that a path $P$ is $\chi$-compliant, given that $P$ is $f$-conforming. Thus, we are equivalently lower bounding the probability that the algorithm succeeds in finding the $f$-conforming path. Let $P=(v_1,e_1,\ldots,e_{r-1},v_r)$. We lower bound the probability that $P$ is half-$\chi$-rainbow. Suppose that the edges in $\bigcup_{i=1}^r f(v_i)$ are colored by $w\leq |\bigcup_{i=1}^r f(v_i)\setminus E(P)|\leq r\mu \leq k\mu$ colors from the set $[q]$. For $P$ to become half-$\chi$-rainbow, we need to ensure that the edges in $E(P)$ are colored by one of the remaining $q-w\geq 4k\mu-k\mu = 3k\mu$ colors and that the coloring is injective on $E(P)$. The probability of that happening for the $r-1$ edges is
    $\frac{q-w}{q}\cdot \frac{q-w-1}{q}\cdot \cdots \cdot \frac{q-w-r+2}{q}$.
    Note that the last term lower bounds every other term and the last term is lower bounded by $\frac{1}{2}$ because $\frac{q-w-r+2}{q}\geq \frac{3k\mu-k}{4k\mu}\geq \frac{1}{2}$ given that $\mu \geq 1$. Hence the probability that $P$ is half-$\chi$-rainbow is at least $2^{-r+1}\geq 2^{-k}$.
     
    Repeat the above for $2^{k}\ln \frac{1}{\varepsilon}$ random choices of $\chi$. Clearly if there is no $f$-conforming $s$-$t$ path the algorithm never finds a $\chi$-compliant path for any $\chi$ due to \Cref{obs:basic_properties_chi_compl_conf_rainbow}. On the other hand, the probability that it outputs no if there is an $f$-conforming path is at most $(1-2^{-k})^{2^{k}\ln\frac{1}{\varepsilon}}\leq (e^{-2^{-k}})^{2^k\ln\frac{1}{\varepsilon}}=e^{-\ln\frac{1}{\varepsilon}}=\varepsilon$. We used the inequality $1-x\leq e^{-x}$ valid for any $x\in\mathbb{R}$ for $x=2^{-k}$. Hence it outputs yes with probability at least $1-\varepsilon$, as desired. Up to multiplicative constants, the running time can be bounded as follows:
    \[
        \binom{4k\mu}{k}2^k(n+m+|f|)\cdot 2^k \ln\frac{1}{\varepsilon}\leq 4^k\left(\frac{4k\mu\cdot e}{k}\right)^k(n+m+|f|)\ln\frac{1}{\varepsilon}\leq 2^{O(k\log \mu)}(n+m)\ln \frac{1}{\varepsilon}
    \]
    Note that $|f|\leq n\cdot \mu$, so $|f|$ gets hidden in the $2^{O(k\log \mu)}\cdot n$ factor in the running time.
    We used the known bound on binomial coefficient: $\binom{n}{k}\leq \left(\frac{n\cdot e}{k}\right)^k$.
\end{proof}
}%
\lv{\subparagraph*{Derandomization.}}
There are two ways to derandomize the \lv{randomized }algorithm given in \Cref{thm:randomized_color_coding_alg}. 
Neither of the two algorithms is Pareto optimal with respect to $\mu$ and $k$.
Hence, each of them turns out to be more suitable in different scenarios. 

\toappendix{\sv{\subsection{Derandomization of the color coding}}}
\toappendix{
\begin{theorem}[Naor et al.~\cite{NaorSS95}]\label{thm:hashing}
    Let $m,q$ be any positive integers.
    There exists a family $\mathcal{F}$ of colorings $\chi\colon [m]\to [q]$ of size $2^{O(q)}\log m$ constructible in $2^{O(q)}m\log m$ time such that for any $F\subseteq [m]$ with $|F|\leq q$ there exists $\chi\in \mathcal{F}$ such that $\chi|_F$ is injective.
\end{theorem}
}
\begin{apprestatable}{theorem}{theoremFptByKandMuWorse}\label{thm:fpt_by_k_and_mu_worse}
    \SSDpath is 
    solvable in $2^{O(\mu k)}(n+m)\log n$ time.
\end{apprestatable}
\toappendix{
\sv{\theoremFptByKandMuWorse*}
\begin{proof}
    Let $q=\mu k+k-1$. We construct in $2^{O(q)}m\log m$ time the family $\mathcal{F}$ from \Cref{thm:hashing}. For every $\chi\in \mathcal{F}$, we decide whether there exists $\chi$-compliant $s$-$t$ path on at most $k$ vertices in $O(2^q(n+m+|f|))$ time using \Cref{lem:chi_compliant_running_time}. Correctness is as follows. Clearly, if we find a $\chi$-compliant $s$-$t$ path on at most $k$ vertices, it is also $f$-conforming by \Cref{obs:basic_properties_chi_compl_conf_rainbow}. On the other hand, if there is an $f$-conforming $s$-$t$ path $P=(v_1,e_1,v_2,\ldots,e_{r-1},v_r)$ for some $r\leq k$, by the properties of $\mathcal{F}$, there is $\chi\in \mathcal{F}$ such that $\chi$ is injective on $F=\bigcup_{i=1}^{r-1}\{e_i\}\cup \bigcup_{i=1}^rf(v_i)$, hence $P$ becomes $\chi$-rainbow, and by \Cref{obs:basic_properties_chi_compl_conf_rainbow}, $P$ is also $\chi$-compliant for this choice of $\chi$, hence it is found by the algorithm. The total running time is $2^{O(q)}m\log m + 2^{O(q)}\log (m) \cdot 2^q(n+m+|f|)$ which is the promised running time $2^{O(\mu k)}(n+m)\log n$. Note again that since $|f|\leq n\cdot \mu$, we hide the $|f|$ factor in the $2^{O(q)}\cdot n$ term in the running time.
\end{proof}
}

\toappendix{
Note that the perfect hashing family from \Cref{thm:hashing} introduces a $2^{O(q)}$ factor in the running time. We can instead use an $(m,q,q^2)$-splitter, but we will need to use $q^2=(\mu k+k-1)^2\leq k^2(\mu+1)^2$ colors.
}

\toappendix{
\begin{definition}[\cite{CyganParaAlg2015}]
    An $(n,k,\ell)$-splitter is a family $\mathcal{F}$ of functions from $[n]$ to $[\ell]$ such that for every set $S\subseteq [n]$ of size $k$ there is a function $\chi\in \mathcal{F}$ that splits $S$ evenly. That is, for every $1\leq j,j'\leq \ell, \Big||f^{-1}(j)\cap S)|-|f^{-1}(j')\cap S|\Big|\leq 1$.
\end{definition}

Note that for $\ell \geq k$ in the above definition, \emph{splitting $S$ evenly} reduces to $f$ being injective on $S$.

\begin{theorem}[\cite{AlonYZ1995}, verbatim from \cite{CyganParaAlg2015}]\label{thm:splitter}
    For any $n,k\geq 1$ one can construct an $(n,k,k^2)$-splitter of size $k^{O(1)}\log n$ in time $k^{O(1)}n\log n$.
\end{theorem}
}
\begin{apprestatable}{theorem}{theoremFptByKAndMuBetter}\label{thm:fpt_by_k_and_mu_better}
    \SSDpath is solvable in $2^{O(k\log(k\mu))}(n+m)\log n$ time.
\end{apprestatable}
\toappendix{
\sv{\theoremFptByKAndMuBetter*}
\begin{proof}
    We use the same approach as in the proof of \Cref{thm:fpt_by_k_and_mu_worse}, but instead of the family of colorings we use an $(m,q,q^2)$-splitter and algorithm with running time $O(\binom{q^2}{k}2^k(n+m+|f|))$ from \Cref{lem:chi_compliant_running_time}. Recall that $q=\mu k+k-1\leq k(\mu+1)$ to obtain the total running time of $q^{O(1)}m\log m + (k^2(\mu+1)^2)^k\cdot 2^k(n+m+|f|)q^{O(1)}\log m \leq 2^{O(k\log (k\mu))}(n+m)\log (n)$, as promised. Here we used the bound $\binom{n}{k}\leq n^k$ and $|f| \in 2^{O(\log \mu)} \cdot n$.
\end{proof}
}%
\subsection{Parameterization by structural parameters and $\mu$ combined.}

We utilize the \FPT algorithms w.r.t. $k$ and $\mu$ combined to obtain several \FPT algorithms for various structural parameters  combined with~$\mu$. In the results that follow, we apply whichever of \Cref{thm:fpt_by_k_and_mu_better} or \Cref{thm:fpt_by_k_and_mu_worse} yields the better asymptotic running time for the given parameters. 

We start with the vertex cover number for which we prove similar bound as in \Cref{lem:length_of_path_in_td_and_vi}.

\begin{apprestatable}{observation}{obsLengthOfPathInVc}\label{obs:length_of_path_in_vc}
    Let $G$ be a graph with vertex cover number $\vcn$. Then $G$ contains no path on more than $2\vcn+1$ vertices.
\end{apprestatable}
\toappendix{
\sv{\obsLengthOfPathInVc*}
\begin{proof}
    Let $S\subseteq G$ be a vertex cover of $G$ of size $\vcn$.
    Let $P$ be any path in $G$ and consider the split of $P$ into segments by $S$. Observe that no segment can contain more than one vertex as otherwise there is an edge in $G\setminus S$, contradicting the assumption that $S$ is a vertex cover. Hence $P$ has at most $2\vcn + 1$ vertices.
\end{proof}
}%

\begin{apprestatable}{corollary}{corollaryFPTByMuAndVcViTd}\label{cor:fpt_by_mu_and_vc_vi_td}
    \SDpath can be solved in \mbox{$2^{O(\mu\cdot \vcn)}(n+m)\log n$} time, in \mbox{$2^{O(\mu\cdot \vi^2)}(n+m)\log n$} time, or in \mbox{$2^{O(\mu\cdot 2^{\td})}(n+m)\log n$} time. Moreover, algorithms with running times 
    \mbox{$2^{o(\vcn)}\poly(n)$}, \mbox{$2^{o(\vi^2)}\poly(n)$}, or \mbox{$2^{2^{o(\td)}}\poly(n)$} for \SDpath even for $\mu=1$ violate ETH.
\end{apprestatable}
\toappendix{
\sv{\corollaryFPTByMuAndVcViTd*}
\begin{proof}
    The algorithms are direct corollaries of \Cref{lem:length_of_path_in_td_and_vi}, \Cref{obs:length_of_path_in_vc}, and  \Cref{thm:fpt_by_k_and_mu_worse}. In fact, if one used \Cref{thm:fpt_by_k_and_mu_better}, we could reduce the $\mu$ factor to $\log \mu$, but an extra $\log \alpha$ factor would appear and this would not give the assymptotically optimal running time for constant $\mu$.
    
    The lower bounds follow from the same ideas as in the proof of \Cref{thm:w1completeness_vi_td}.
\end{proof}
}%
Unless $\FPT = \W{1}$, the \FPT algorithm for vertex cover cannot be extended already to the parameter distance to linear forest. 

\begin{apprestatable}{theorem}{theoremWOneCompletenessDtlfMuOne}\label{thm:w1completeness_dtlf_mu1}
    \SDpath is \W{1}-complete parameterized by the distance to linear forest, even if $\mu\leq 1$.
\end{apprestatable}

\toappendix{
\sv{\theoremWOneCompletenessDtlfMuOne*}
\begin{proof}
    Membership in \W{1} follows from \W{1} membership for \fvsn (\Cref{thm:membership_w1_fvsn}). For the hardness, take \Cref{construction:mcc} and replace each vertex $y_w$ by a path on $|f(y_w)|$ vertices and delete one edge of the original set $f(y_w)$ per vertex of the new path. Such graph has $\mu_f\leq 1$. The graph $G\setminus \{g_0,\ldots,g_k\}$ is not not necessarily an independent set but a collection of vertex-disjoint paths. Hence $G$ has distance to disjoint paths at most $k+1$.
\end{proof}
}%
\lv{
Notice that with the modification of \Cref{construction:mcc} from the proof of \Cref{thm:w1completeness_dtlf_mu1}, the resulting path might be long. This confirms the intuition that if we parameterize by both $k$ and $\mu$, \SSDpath is \FPT, while parameterization by $k$ or $\mu$ alone is \W{1}-hard (resp. \paraNP-hard).
}
\subsection*{Dense parameters and $\mu$ combined.}

In the remainder of this section we focus on parameters whose bounded values together with the existence of a long path imply dense structure of the graph in some sense. As a warmup example, recall that \SDpath is \NP-hard on cliques (\Cref{cor:hardness_cliques}). Observe that there is always an $f$-conforming path on at most $\mu+2$ vertices when the underlying graph is a clique, if there is any $f$-conforming path at all. This is because the vertex $s$ cannot delete more than $\mu$ edges, so if the path is longer, then there is a shortcut that we can take. By plugging this upper bound into \Cref{thm:fpt_by_k_and_mu_better}, we obtain the following:

\begin{observation}\label{obs:fpt_on_cliques_worse}
\SDpath is solvable in $2^{O(\mu \log \mu)}(n+m)\log n$ time on cliques.    
\end{observation}
We build upon this idea. In general, vertices of a $k$-vertex path can delete at most $k\mu$ edges. On the other hand, if we consider only shortest ($f$-conforming) paths, any shortcut edge on the vertices of the path must be deleted as otherwise the path can be shortened, contradicting that it is the shortest path. We thus obtain the following:

\begin{observation}\label{lem:shortest_path}
    Let $(G=(V,E),f)$ be a self-deleting graph and let $P=(v_1,v_2,\ldots,v_k)$ be a shortest $f$-conforming $v_1$-$v_k$ path in $(G,f)$. Let $G_P=G[V(P)]$ be the subgraph of $G$ induced by the vertices of $P$. Then $|E(G_P)|\leq k\mu+k-1\leq k(\mu+1)$.
\end{observation}

\renewcommand{\arraystretch}{1.4}
\begin{table}[th!]
    \centering
    \begin{tabular}{|c|c|c|c|}
        \toprule
         parameter $\alpha$ & lower bound & length of path & running time \\\midrule
         distance to cograph & $n\log \frac{n}{\alpha}$ & $\alpha2^{O(\mu)}$ & $2^{\alpha2^{O(\mu)}}(n+m)\log n$ \\\hline
         cluster vertex deletion number
        & $\frac{n^2}{\alpha}$ & $O(\alpha\mu)$ & $2^{\widetilde{O}(\alpha \mu)}(n+m)\log n$ \\\hline
         neighborhood diversity & $\frac{n^2}{\alpha}$ & $O(\alpha \mu)$ & $2^{\widetilde{O}(\alpha \mu)}(n+m)\log n$ \\\hline
         modular-width & $n\log_{\alpha}n$ & $\alpha^{O(\mu)}$ & $2^{\alpha^{O(\mu)}}(n+m)\log n$\\\hline
         maximum induced matching & $\frac{n^2}{\alpha}$ & $O(\alpha \mu)$ & $2^{\widetilde{O}(\alpha \mu)}(n+m)\log n$  \\\hline
         shrub-depth & $n^{1+\frac{1}{2^\alpha-1}}$ & $\mu^{2^{O(\alpha)}}$ & $2^{\mu^{2^{O(\alpha)}}}(n+m)\log n$ \\\bottomrule
    \end{tabular}
    \caption{
    Overview of the framework for parameterization by structural parameters $\alpha$ and $\mu$ combined. The lower bound column is the asymptotic lower bound on the number of edges for traceable graph with given parameter bounded by $\alpha$ proven by Dvořák et al.~\cite{DvorakKOPSS2025}. The third column indicates what is the implied upper bound on the length of any shortest $f$-conforming path in such graphs. The fourth column is the final running time of the algorithm using the better from \Cref{thm:fpt_by_k_and_mu_better,thm:fpt_by_k_and_mu_worse}. Note that $\alpha$ should be replaced by $\max \{\alpha, 1\}$. We write just $\alpha$ for readability purposes.}
    \label{tab:comprehensive_overview_c_plus_structural}
\end{table}

We now utilize the results of Dvořák et al.~\cite{DvorakKOPSS2025}. They provide a lower bound on the number of edges in the input graph, given that it contains long (Hamiltonian) path. Recall that graphs containing a Hamiltonian path are also called \emph{traceable}. By combining these bounds together with \Cref{lem:shortest_path}, we obtain an upper bound on the length of a shortest $f$-conforming $s$-$t$ path in the underlying graph in terms of $\mu$ and some structural parameter $\alpha$, see \Cref{tab:comprehensive_overview_c_plus_structural} for an overview. Recall that $\widetilde{O}$ supresses polylogarithmic factors.

\begin{lemma}\label{lem:ck_parameter}
    Let $\alpha$ be a graph parameter monotone under taking induced subgraphs. Suppose there are global constants $C,D$ such that any traceable graph $G$ with $n\geq C\cdot \alpha(G)$ vertices contains at least $D\frac{n^2}{\alpha(G)}$ edges. Then \SDpath can be solved in $2^{\widetilde{O}(\alpha \mu)}(n+m)\log n$ time. In particular it is \FPT parameterized by $\alpha$ and $\mu$ combined.
\end{lemma}

\begin{proof}
    Let $(G,f,s,t)$ be the input instance of \SDpath and let $P$ be a shortest $f$-conforming $s$-$t$ path in $G$. Let $k=|V(P)|$ and let $G_P=G[V(P)]$ be the graph induced by vertices of $P$. By assumption on $\alpha$, either $k<C \cdot \alpha(G_P)\leq C \cdot \alpha(G)$ (because $\alpha$ is monotone),
    or $|E(G_P)|\geq D\frac{k^2}{\alpha(G_P)}$. By \Cref{lem:shortest_path} we also have $|E(G_P)|\leq k(\mu+1)$. By combining these two bounds we obtain the bound $k\leq \frac{1}{D}\alpha(G_P)(\mu+1) \leq \frac{1}{D}\alpha(G)(\mu+1)$. By plugging $k=\max\{C\cdot \alpha(G),\frac{1}{D}\alpha(G)\cdot (\mu+1)\}$ into \Cref{thm:fpt_by_k_and_mu_better} we obtain the promised algorithm with running time $2^{O(\alpha \mu \log (\alpha \mu))}(n+m)\log n$.
\end{proof}

We could prove an analogous version of \Cref{lem:ck_parameter} also with the functions $n\mapsto D\cdot n\log \frac{n}{\alpha}$, $n\mapsto D\cdot n\log_\alpha n$, or $n\mapsto D\cdot n^{1+\frac{1}{2^{\alpha}-1}}$ with different running time of the algorithm (see \Cref{tab:comprehensive_overview_c_plus_structural}).

\begin{theorem}\label{thm:fpt_by_dense_and_mu}
    \SDpath is \FPT w.r.t. $\alpha$ and $\mu$ combined where $\alpha$ is one of the following parameters: cluster vertex deletion number, neighborhood diversity, distance to cograph, modular-width, maximum induced matching, shrub-depth.
\end{theorem}
\begin{proof}
    We prove the theorem for cluster vertex deletion number, the rest is proved similarly by using suitable lower bound from \Cref{tab:comprehensive_overview_c_plus_structural}. If $\cvdn(G)=0$, the graph is a cluster and we can restrict ourselves to the clique where $s$ and $t$ lies and use \Cref{obs:fpt_on_cliques_worse}. Otherwise, if $\cvdn(G)\geq 1$, by the result of Dvořák et al.~\cite{DvorakKOPSS2025} any traceable graph $G$ on $n \geq 4\cvdn(G)$ vertices contains at least $\frac{n^2}{16(\cvdn(G)+1)}\geq \frac{n^2}{32\cvdn(G)}$ edges. Invoke \Cref{lem:ck_parameter} for $D=\frac{1}{32}$ and $C=4$ to obtain the desired algorithm.
\end{proof}

\subparagraph*{Domination Number}
\FPT algorithms for distance to cograph, modular-width, or maximum induced matching cannot be extended to an \FPT algorithm for the diameter of the graph (and $\mu$ combined). We show that \SDpath is \paraNP-hard already for domination number and $\mu$ combined.

\begin{apprestatable}{theorem}{theoremHardnessDomination}\label{thm:hardness_domination}
       \SDpath remains \NP-hard even when the self-deleting graph $(G,f)$ satisfies $\gamma(G)=\mu_f=1$.
\end{apprestatable}
\toappendix{
\sv{\theoremHardnessDomination*}
\begin{proof}
We reduce from \SDpath with $\mu \leq 1$ which is \NP-hard by \Cref{cor:hardness_fv1}. Let $(G,f,s,t)$ be such an instance of \SDpath.
    
    We create a new graph $G'$ from $G$ by attaching leaves $s'$ and $t'$ to $s$ and $t$, respectively, and then adding a universal vertex $u$. 
    Furthermore, we construct $f'$ in such a way that $f'(v) = f(v)$ for each vertex $v \in V(G)$, $f'(u)=\{\{t,t'\}\}$, and $f(s') = \{\{u, t'\}\}$. 
    The new instance is $(G',f',s',t')$. 
    Clearly, $\{u\}$ is a dominating set for $G'$, hence $\gamma=1$ and $\mu=1$, as claimed. 

    We now show that there is an $f$-conforming $s$-$t'$ path $P$ in $G$ if and only if there is an $f'$-conforming $s'$-$t'$ path $P'$ in $G'$.
    Both edges $\{s', s\}$ and $\{t,t'\}$ are not deleted by any vertex in $G$ nor $s'$ nor $t'$. 
    Therefore, $P$ can be extended by prepending $s'$ and appending $t'$ to create $f'$-conforming $s'$-$t'$ path in $G'$.

    In the other direction the vertex $t'$ is connected to the rest of $G'$ with edges $\{t, t'\}$ and $\{u, t'\}$.
    The edge $\{u, t'\}$ is immediately removed by $s'$.
    Including $u$ in $P$ disconnects the vertex $t'$ from the rest of $G'$ and thus $u$ is not in $P$.
    Therefore, $P$ can be obtained from $P'$ by removing the first and last vertex.      
\end{proof}
}%

\appsection{Kernels}{sec:kernels}
\label{sec:kernel}
In this section, we show that while \SDpath admits \FPT algorithms for broad number of structural parameters with $\mu$ combined, it does not admit a polynomial kernel w.r.t.\  the vertex cover number and $\mu$ combined already in the class of $2$-outerplanar graphs unless the polynomial hierarchy collapses. Moreover, there is no polynomial kernel in the class of cliques w.r.t.\ $\mu$, but there is a linear Turing kernel w.r.t.\ $\mu$ in the class of cliques.

\toappendix{
\subsubsection*{Kernelization Preliminaries}
A \emph{kernel} for a parameterized problem $Q$ is an algorithm that, given an instance $(x,k)$ of~$Q$, works in polynomial time in $(|x|+k)$ and returns an equivalent instance $(x',k')$ of $Q$. Moreover, $|x'|+k'\leq g(k)$ for some computable function $g$.
A \emph{Turing kernel} for a parameterized problem~$Q$ is an algorithm that, given an instance $(x,k)$ of~$Q$, decides whether $(x,k)\in Q$ in time polynomial in $(|x|+k)$, when given access to an oracle that decides membership in~$Q$ for any instance $(x',k')$ with $|x'|+k'\leq g(k)$ in a single step for some computable function $g$.
In the above definitions, if $g$ is polynomial or linear function, we say that $Q$ admits a \emph{polynomial or linear (Turing) kernel}, respectively.
\begin{definition}[polynomial equivalence relation~\cite{BodlaenderJK2012_nonkernels}]
An equivalence relation~$\mathcal{R}$ on the set~$\Sigma^*$ is called a \emph{polynomial equivalence relation} if the following conditions are satisfied:
\begin{enumerate}
    \item There exists an algorithm that, given strings $x,y\in \Sigma^*$, decides whether $(x,y)\in \mathcal{R}$ in time polynomial in $|x|+|y|$.
    \item For every $n\in\mathbb{N}$, $\mathcal{R}$ splits the set of strings from $\Sigma^*$ of length at most $n$ into at most $\operatorname{poly}(n)$ many equivalence classes.
\end{enumerate}
\end{definition}

\begin{definition}[OR-cross-composition~\cite{BodlaenderJK2012_nonkernels}]
    Let $L\subseteq \Sigma^*$ be a language and $\mathcal{R}$ a polynomial equivalence relation on $\Sigma^*$ and let $Q\subseteq \Sigma^*\times \mathbb{N}$ be a parameterized problem. An \emph{OR-cross-composition of $L$ into $Q$ (with respect to $\mathcal{R}$)} is an algorithm that, given $\tau$ instances $x_1,x_2,\ldots,x_{\tau}\in \Sigma^*$ of $L$ belonging to the same equivalence class of $\mathcal{R}$, takes time polynomial in $\sum_{i=1}^{\tau}|x_i|$ and outputs an instance $(y,k)\in \Sigma^*\times \mathbb{N}$ such that:
    \begin{enumerate}
        \item $k\leq \operatorname{poly}(\max_{i}|x_i|+\log \tau)$
        \item $(y,k)\in Q$ if and only if there is $i\in[\tau]$ such that $x_i\in L$.
    \end{enumerate}
\end{definition}

\begin{theorem}[\cite{BodlaenderJK2012_nonkernels}]\label{thm:cross_composition_machinery}
    If an \NP-hard language $L$ OR-cross-composes into the parameterized problem $Q$, then $Q$ does not admit a polynomial kernel unless $\NP\subseteq \coNP/_{\operatorname{poly}}$.
\end{theorem}

We remark that $\NP\subseteq \textsf{coNP}/_{\operatorname{poly}}$ implies that the polynomial hierarchy collapses to the third level~\cite{Yap1983}. %
}

\begin{apprestatable}{theorem}{theoremNoKernel}\label{thm:nokernel}
     Unless $\NP\subseteq \coNP/_{\operatorname{poly}}$, \SDpath does not admit a polynomial kernel with respect to 
     \begin{enumerate}[a)]
      \item $\vcn$ and $\mu$ combined, even on $2$-outerplanar graphs;
      \item $\vi$ even with $\mu=1$ and on $2$-outerplanar graphs;
      \item $\mu$ on cliques.
     \end{enumerate}
\end{apprestatable}
\sv{
\begin{proof}[Proof Sketch for \itemstyle{a)}]
    We provide an OR-cross-composition of \tsat into \SDpath parameterized by $\vcn$ and $\mu$. The idea is to use \Cref{construction:sat} for a formula with all $O(n^3)$ clauses with auxiliary \emph{skip edges} that allow to pass a clause for free. Before $s$ we prepend selector vertices for $\tau$ instances that will delete appropriate skip edges and thus modify the rest of the graph to look like the reduction for the given formula. Details can be found in the appendix.
\end{proof}

}

\toappendix{
\sv{\theoremNoKernel*}
\begin{proof}
\begin{enumerate}[a)]
    \item We provide an OR-cross-composition of \tsat into \SDpath parameterized by $\vcn$ and $\mu$. We first define a suitable polynomial equivalence relation $\mathcal{R}$. Two instances of \tsat are $\mathcal{R}$-equivalent if they contain the same number of variables. Given $\tau$ formulas $\varphi_1,\varphi_2,\ldots,\varphi_{\tau}$ over $n$ variables (without loss of generality) $X=\{x_1,x_2,\ldots,x_{n}\}$, invoke \Cref{construction:sat} for the formula $\varphi^*$ that contains all $\binom{2n}{3}\leq 8n^3$ possible clauses on~$X$. Slightly modify the clause gadget for each clause $C\in \varphi^*$ by adding a \emph{skip edge} from $\iota_{C}$ to $o_C$. Denote this part of the construction as $G$. Next, add a new starting vertex $s'$ and connect $s'$ with $s$ by $\tau$ edge-disjoint paths of length $2$. Denote the $i$-th path by $P_i=(s',v_i,s)$. The middle vertices $v_i$ will be responsible for selecting the instance to activate inside $G$. More precisely, the vertex $v_i$ deletes all skip edges corresponding to clauses that are present in the formula $\varphi_i$. Formally, $f(v_i)=\{\{\iota_C,o_C\}\mid C\in \varphi_i\}$. The resulting graph is denoted~$G'$, the starting vertex is $s'$, and the target vertex is $t$ (as in~$G$). The correctness of the construction inside $G$ follows from similar arguments as in \Cref{claim:construction_sat_correctness}. The selectors~$v_i$ guarantee that exactly the clauses $C\in \varphi_i$ must be passed using the edges $e_\ell^C$ and the remaining clauses can be skipped using the skip edges. Hence, $(G',f,s',t)$ is a yes-instance if and only if at least one of $\varphi_1,\ldots,\varphi_{\tau}$ is satisfiable. Note that size of $G$ is polynomial in $\max_{i}|\varphi_i|$, the set $V(G)\cup \{s'\}$ forms a vertex cover for $G'$ and $\mu \leq 8n^3$, hence this is a valid OR-cross-composition. It remains to argue that $G'$ is $2$-outerplanar. Clearly we can draw $G'$ by having the vertex $s'$ and all the $\iota$ and $o$ vertices of all the gadgets on the outer face. Removing these vertices leaves a forest which is clearly outerplanar. The claimed result follows from \Cref{thm:cross_composition_machinery}.
    \item We use the same approach as in \itemstyle{a)}, but modify \Cref{construction:sat} as in \Cref{cor:hardness_fv1}: replace each selector vertex $v_i$ by a path of length $|f(v_i)|$. Now the set $V(G)\cup \{s'\}$ is not necessarily a vertex cover but leaves a set of paths that are polynomial in the size of $G$ (which is polynomial in $\max_{i}|\varphi_i|$). Hence $\vi(G')$ is polynomial in $\max_{i}|\varphi_i|$ and thus this is a valid OR-cross-composition of \tsat into \SDpath parameterized by $\vi$ with $\mu = 1$.
    \item We again use the same approach as in \itemstyle{a)}, but add all missing edges so that $G'$ is a clique. Now, each vertex of $G'$ only deletes its nonexistent edges into $G$ (but not into the selector vertices $v_i$). Any $f$-conforming $s'$-$t$ path is still forced to take at least one selector vertex $v_i$ before entering $G$ and there is clearly no advantage of visiting a selector vertex twice. More precisely, every $f$-conforming path visiting more than one selector vertex can be shortened to visit exactly one selector vertex.
    \end{enumerate}
\end{proof}
}%

\begin{apprestatable}{theorem}{theoremKernel}\label{thm:kernel}
    \SDpath admits 
    \begin{enumerate}[a)]
     \item an $O(\fen)$ kernel;
     \item an $O(\vcn)$ kernel on outerplanar graphs;
     \item a Turing kernel with $O(\mu)$ vertices on cliques.
    \end{enumerate}    
\end{apprestatable}

%
%
%

\toappendix{
\sv{\theoremKernel*}
\sv{\subparagraph*{Linear kernel for feedback edge number}}
To prove \Cref{thm:kernel} \itemstyle{a)}, we design two reduction rules. Without loss of generality, suppose that the input graph $G$ is connected and let $F\subseteq E(G)$ be the feedback edge set of size $\fen(G)$. Note that $F$ can be found in linear time by finding the spanning tree of $G$. We let $T=G\setminus F$.

\begin{rrule}\label{rrule:fen_leaves}
    If $v\in V(G)\setminus \{s,t\}$ is a leaf in $G$, remove it.
\end{rrule}

\begin{rrule}\label{rrule:fen_degree2}
    Suppose there is a sequence $v_1,v_2,\ldots,v_r$ of vertices for $r\geq 4$ with the following properties:
    \begin{enumerate}[a)]
        \item $s,t \notin W = \{v_2,v_3,\ldots,v_{r-1}\}$,
        \item vertices in $W$ are not incident to an edge of $F$,
        \item vertices $v_1,v_2,v_3,\ldots,v_{r-1},v_r$ form a path in $G$, and 
        \item vertices in $W$ have degree $2$ in $G$.
    \end{enumerate}
    Then do the following:
    \begin{itemize}
    \item Delete all vertices of $W$ from $G$.
    \item If $v_1 \neq v_r$, then add a new vertex $v^*$ with neighbors $v_1$ and $v_r$ and modify the deletion sets as follows. 
        \begin{enumerate}[i)]
            \item Vertex $v^*$ deletes all edges previously deleted by vertices of $W$ not incident to any of them.
            \item If a vertex $v \notin W$ deletes some edge $\{v_i,v_{i+1}\}$, then $v$ now deletes both edges $\{v_1,v^*\}$ and $\{v^*,v_r\}$.
            \item If the path $(v_1,v_2,\ldots,v_r)$ is not $f$-conforming, then $v^*$ also deletes the edge $\{v^*,v_r\}$. 
            If the path $(v_r,v_{r-1},\ldots,v_1)$ is not $f$-conforming, then $v^*$ also deletes the edge $\{v_1,v^*\}$.
        \end{enumerate}
    \end{itemize}
\end{rrule}

Correctness of \Cref{rrule:fen_leaves} immediately follows because no leaves of $G$ except for $s$ or $t$ can be part of any $s$-$t$ path in $G$.

\begin{lemma}
    \Cref{rrule:fen_degree2} is correct.
\end{lemma}
\begin{proof}
    Denote by $G$ the old graph and by $G'$ the new graph obtained from $G$ by applying \Cref{rrule:fen_degree2} once. Let $P=(s=u_1,e_1,\ldots,u_k=t)$ be an $f$-conforming $s$-$t$ path in $G$ and suppose that $P$ was touched by the application of \Cref{rrule:fen_degree2}. Hence for some $i\geq 2$ and $j = i+r -3$ we have $W=\{u_i,u_{i+1},\ldots,u_{j}\}$. The new path is created by replacing this segment by the vertex $v^*$. Denote the new path $P'$. Path $P'$ is $f$-conforming because $v^*$ deletes exactly what the vertices in $W$ deleted and the edges $\{u_{i-1},v^*\}$ and $\{v^*,u_{j}\}$ are not deleted because no vertex before $u_i$ deletes any edge inside $W$. Hence $P'$ is $f$-conforming.

    On the other hand, suppose that $P'=(s=u_1,e_1,\ldots,u_k=t)$ is an $f$-conforming path in~$G'$ and suppose that $u_i=v^*$ for some~$i$. Let $P$ be a path in $G$ created from $P'$ by replacing~$v^*$ by the vertices in~$W$. No edge incident to $W$ on $P$ is deleted by some vertex $u$ before $v^*$ as otherwise $u$ also deleted both edges incident to $v^*$, contradicting the assumption that~$P'$ is $f$-conforming. Finally, no vertex of $W$ deletes any edge after the vertex $v^*$ as otherwise also $v^*$ deleted them, again contradicting the assumption that $P'$ is $f$-conforming.
\end{proof}

The kernelization algorithm consists of exhaustively applying \Cref{rrule:fen_leaves,rrule:fen_degree2}. It remains to show that after their exhaustive application the total number of vertices is linear in $|F|$. To do so, we upper bound the number of vertices of degrees $1,2$ and $\geq 3$ in $T$. First, we recall a well known fact about the number of inner vertices of degree at least $3$ in a tree.

\begin{observation}\label{obs:bound_on_vertices_of_atleast3}
    Let $T$ be any tree and $\ell$ the number of leaves of $T$. Then the number of vertices of degree at least $3$ is at most $\ell-2$.
\end{observation}
\begin{proof}
    Recall that trees have $n-1$ edges and the handshaking lemma gives $\sum_{v\in V}\deg (v) = 2n-2$. Let $d_1,d_2,d_{\geq 3}$ denote the number of vertices of degrees $1$, $2$ and $\geq 3$ in $T$, respectively. Then we have $2(d_1+d_2+d_{\geq 3})-2=2n-2=\sum_{v\in V}\deg v \geq d_1+2d_2+3d_{\geq 3}$. Thus $d_{\geq 3}\leq d_1-2=\ell-2$.
\end{proof}

\begin{lemma}
    If $G$ is reduced with respect to \Cref{rrule:fen_leaves,rrule:fen_degree2}, then $|V(G)|\leq 8\fen+4$.
\end{lemma}
\begin{proof}
    Note that there are at most $2\fen$ vertices of $G$ incident to an edge of $F$. The remaining vertices satisfy $\deg_T(v)=\deg_G(v)$. By assumption that \Cref{rrule:fen_leaves} was applied exhaustively, leaves of $T$ are either the terminal vertices $s,t$ or vertices incident to an edge of $F$. Hence there are at most $2\fen+2$ leaves in $T$. Next, by \Cref{obs:bound_on_vertices_of_atleast3} there are at most $(2\fen + 2) - 2 = 2\fen$ vertices of degree at least $3$ in $T$. Finally, we bound the number of vertices of degree $2$. To do this, we root $T$ in an arbitrary vertex $r$. Note that we applied \Cref{rrule:fen_degree2} exhaustively, so every vertex of degree $2$ is either $s$ or $t$ or it is adjacent only to vertices of degree at least $3$, vertices incident with an edge of $F$, or to $s$ or $t$.
    To bound number of vertices of degree $2$ consider mapping each vertex of degree $2$ in $T$ to its child. As argued above, each vertex of degree $2$ except possibly $s$ or $t$ is mapped to a vertex of degree at least $3$ or to a vertex incident to an edge of $F$, or $s$ or $t$. As this mapping is clearly injective, we obtain that there are at most $2\fen + 2 + 2\fen=4\fen + 4$ of vertices of degree $2$ in $T$ (note that vertices $s$ and $t$ were already counted). Hence, the total number of vertices in $T$ (and hence in $G$) is at most $8\fen + 4$, which proves the lemma.
\end{proof}
This finishes the proof \Cref{thm:kernel} \itemstyle{a)}.
}%

\toappendix{
\begin{proof}[Proof of \Cref{thm:kernel} \itemstyle{b)}]
     Let $G$ be the input graph. Towards the kernel, we utilize just \Cref{rrule:fen_leaves}. We claim that the resulting graph has $O(\vcn)$ vertices and edges. Let $G'$ be created from $G$ by exhaustively applying \Cref{rrule:fen_leaves} and let $S\subseteq V(G')$ be a vertex cover for $G'$ of size $\vcn=\vcn(G')$. We bound the number of vertices $v\in V(G')\setminus S$ according to size of $N(v)\cap S$. Note that there are no edges with both endpoints in $G'\setminus S$ as $S$ is a vertex cover. Let $D_1,D_2,D_{\geq 3}$ be the sets of vertices in $G'\setminus S$ of degrees $1,2$ and $\geq 3$, respectively. Clearly $|D_1|\leq 2$ as the only remaining vertices of degree $1$ are $s$ or $t$ because $G'$ is reduced with respect to \Cref{rrule:fen_leaves}.
    \begin{claim}
        $|D_{\geq 3}|\leq 2|S|-3$.
    \end{claim}
    \begin{claimproof}
        Let $D=D_{\geq 3}$.
        Consider the induced subgraph $H=G'[S\cup   D]$. Since $H$ is outerplanar, $|E(H)|\leq 2|V(H)|-3= 2(|D|+|S|)-3$. On the other hand, $H$ has at least $3|D|$ edges. We obtain $2(|D|+|S|)-3\geq 3|D|$. This yields $|D|\leq 2|S|-3$.
    \end{claimproof}
    \begin{claim}
        $|D_2|\leq 4|S|-6$
    \end{claim}
    \begin{claimproof}
        Let $D=D_2$. Consider the subgraph $H=G'[S\cup D]$. For the sake of bounding the size of $D$, suppose that we contract each vertex of $D$ to a single edge. More precisely, given $v\in D$ and $N(v)=\{u_1,u_2\}$, we remove the vertex $v$ and add an edge $\{u_1,u_2\}$. We obtain an outerplanar multigraph $H_1$. Note that $H_1$ may contain parallel edges. However, note that there cannot be three (or more) parallel edges between two vertices as otherwise the input graph contained $K_{2,3}$ as a subgraph, contradicting the fact that $H$ is outerplanar. Let $H_2$ be the simple graph that results from $H_1$ by removing parallel edges. We obtain $|E(H_2)|\leq 2|V(H_2)|-3$. By the above argumentation we have $|E(H_1)|\leq 2|E(H_2)|$. But each edge of $H_1$ corresponds to a vertex of $D$. Finally, note that $|V(H_1)|=|V(H_2)|=|S|$. Thus we obtain $|D|\leq |E(H_1)|\leq 2|E(H_2)| \leq 2(2|V(H_2)|-3)=4|S|-6$.
    \end{claimproof}

    Hence, the graph $G'$ has at most $|S|+|D_1|+|D_2|+|D_{\geq 3}|\leq |S| + 2 + 4|S|-6 + 2|S|-3 \leq 7|S|= 7\vcn(G')\leq 7\vcn(G)$ vertices. The last inequality follows from the fact that $G'$ is a subgraph of $G$. Since $G'$ is itself outerplanar, it has at most $14\vcn(G')$ edges and this finishes the proof. 
 \end{proof}
 \begin{proof}[Proof of \Cref{thm:kernel} \itemstyle{c)}]
     Let $(G=(V,E),f,s,t)$ be the input instance. Let $P=(s=v_1,e_1,v_2,\ldots,e_{r-1},v_r=t)$ be a shortest $f$-conforming $s$-$t$ path in $G$. Observe that the vertex $s$ must delete all the edges of the form $\{s,v_i\}$ for $v_i\in\{v_3,v_4,\ldots,v_r\}$. Hence, we only guess the second vertex $v_2$ on the path and the remaining vertices must be those that have their edge to $s$ deleted by $s$. For $v\in V\setminus \{s\}$, we let $X^v=\{s,t,v\}\cup \{u\in V(G) \mid \{s,u\}\in f(s)\}$. We create $n-1$ instances of the form $(G[X^v],s,t,f')$, where $f'$ is restriction of $f$ to $G[X^v]$, for each $v\in V$. Note that $|X^v|\leq \mu + 3$ and this is the claimed linear Turing kernel. There is an $f$-conforming $s$-$t$ path in $G$ if and only if there is $v\in V\setminus \{s\}$ such that there is an $f$-conforming $s$-$t$ path in $G[X^v]$.   
\end{proof}

Note that for the proof we only used that the vertex $s$ is universal in $G$.
Hence, in fact the theorem yields a Turing kernel with $O((n - \deg s)+|f(s)|)$ vertices in general graphs.
}%

\begin{apprestatable}{corollary}{corollaryFptOnCliquesSingleExp}\label{cor:fpt_on_cliques_singleexp}
    \SDpath can be solved in $2^{O(\mu)}n^2$ time if the underlying graph is a clique and $2^{o(\mu)}\poly(n)$-time algorithm on cliques violates ETH.
\end{apprestatable}
\toappendix
{
\sv{\corollaryFptOnCliquesSingleExp*}
Before proving \Cref{cor:fpt_on_cliques_singleexp}, we design a single-exponential algorithm for \SDpath on general self-deleting graphs. 
Let $(G=(V,E),f)$ be a self-deleting graph and let $\mathcal{D}(G)=\{f(v)\mid v \in V\}$ be the set of distinct deletion sets in $G$. We present an algorithm with running time $O(2^{|\mathcal{D}(G)|}(n+m+|f|))$. Note that any self-deleting graph satisfies $|\mathcal{D}(G)|\leq n$, hence the running time of the algorithm can also be expressed as $O(2^n(n+m+|f|))$.

\begin{theorem}\label{thm:single_exp_alg}
    \SDpath can be solved in $O(2^{|\mathcal{D}(G)|}(n+m+|f|))$ time. Moreover, \SDpath cannot be solved in  $2^{o(|\mathcal{D}(G)|})\poly(n)$ time unless ETH fails.
\end{theorem}
\begin{proof}
    Let $\mathcal{D}(G)=\{D_1,\ldots,D_k\}$. For vertex $v$, let $\operatorname{type}(v)$ denote the index $i$ such that $f(v)=D_i$.
    Construct a new graph $G'$ by creating $2^k$ copies of $G$ as follows. For simplicity, regard $G$ as a directed graph where an edge $\{u,v\}$ is represented by a pair of edges $(u,v),(v,u)$. If $\{u,v\}\in f(w)$, then both $(u,v)$ and $(v,u)$ are in $f(w)$.
    Formally, we construct a new directed graph $G'$, where $V(G')=\{(v,S)\mid S\subseteq [k],v\in V(G)\}$. Now, for each edge $(u,v)\in E(G)$ in the original graph, for each $S\subseteq [k]$ such that $(u,v)\notin \bigcup_{i\in S}D_i$, we add an edge from $(u,S)$ to $(v,S\cup \{\operatorname{type}(v)\})$ into $E(G')$.

    Each layer of $G'$ represents the graph, where we the set of edges $\bigcup_{i\in S}D_i$ is deleted. Every time we visit a vertex $v$ (i.e., we use edge $(u,v)$), we reflect the actual deleting set. We can now run a simple breadth-first search on $G'$ and look for a simple path from $s' = (s,\{\operatorname{type}(s)\})$ to the vertex $t'=(t,S)$ for any $S\subseteq [k]$. It is not hard to verify that $f$-conforming $s$-$t$ paths in~$G$ correspond to $s'$-$t'$ paths in~$G'$. Note that we can build $G'$ in $O(2^k(n+m+|f|))$ time and its size is $2^k(n+m)$.
    
    For the lower bound, consider \Cref{construction:sat} and notice that $|\mathcal{D}(G)|=2n'+1$, where $n'$ is the number of variables of the \textsc{$3$-Sat} formula and $G$ is the resulting graph of the reduction. Hence an $2^{o(|\mathcal{D}(G)|)}\poly(n)$ algorithm for \SDpath would yield $2^{o(n')}$ algorithm for \textsc{$3$-Sat}, contradicting ETH.
\end{proof}

\begin{proof}[Proof of \Cref{cor:fpt_on_cliques_singleexp}]
    For the algorithm we use \Cref{thm:kernel} \itemstyle{c)} together with \Cref{thm:single_exp_alg} (note that $|f| \in 2^{O(\mu)}\cdot n$). The lower bound follows by chaining \Cref{construction:sat} with the modification from the proof of \Cref{cor:hardness_fv1} and then applying the reduction from the proof of \Cref{cor:hardness_cliques}.
\end{proof}
}

\section{Conclusion and open problems}
We initiated a systematic study of complexity of finding a simple path in a self-deleting graph, which we call \SDpath. While the problem is hard on very restricted graph classes, we were able to design \FPT algorithm(s) parameterized by the solution size and structure of the deletion function. This further allowed us to design \FPT algorithms for various structural parameters combined with the structure of the deletion function. 

\lv{\subparagraph*{Better derandomization of color-coding.}}
Our derandomization of the color-coding algorithm yields running time of $2^{O(k\log (k\mu))}\poly(n)$ (\Cref{thm:fpt_by_k_and_mu_better}) or $2^{O(k\mu)}\poly(n)$ (\Cref{thm:fpt_by_k_and_mu_worse}). We were unable to derandomize it in a way to match the randomized running time of $2^{O(k\log \mu)}\poly(n)$ from \Cref{thm:randomized_color_coding_alg}. We conjecture that with a suitable pseudorandom object, there is a way to derandomize the algorithm into a deterministic $2^{O(k\log \mu)}\poly(n)$ time.

\lv{\subparagraph*{Optimal algorithms for structural parameters and $\mu$ combined.}}
Our framework for \FPT algorithms parameterized by $k$ and $\mu$ together with lower bounds on the number of edges in traceable graphs with dense structure does not give optimal running times under ETH. For example, already for cliques, the framework gives running time $2^{O(\mu \log \mu)}(n+m)$ (\Cref{obs:fpt_on_cliques_worse}) which is not optimal (\Cref{cor:fpt_on_cliques_singleexp}). 
Assuming ETH, we cannot obtain an algorithm for \SSDpath with running time $2^{o(k)\log \mu}\poly(n)$ (see \Cref{rem:sat_is_linear}). 
Similarly, an algorithm with running time $2^{k \cdot o(\log \mu)}\poly(n)$ would imply that \SSDpath parameterized by $k$ is in \FPT (this is due to~\cite[Lemma 1]{CaiJ01subexpcollapse}),
which would then imply that $\FPT=\W{1}$. Note that \SSDpath becomes \FPT w.r.t. $k$ if $\mu \in O(\log n)$ by plugging into the algorithm from \Cref{thm:fpt_by_k_and_mu_better} and using the fact that $(\log n)^{g(k)}$ is fpt-time.

\lv
{
\todo[inline]{But it does not admit even a Turing polynomial kernel for the parameter $k$ with $\mu = O(\log n)$}
}

Can the algorithms for structural parameters and $\mu$ be improved to match the above ETH lower bounds? For example, is it possible to solve \SDpath in deterministic $2^{O(\vcn \log \mu)}\poly(n)$ time (note that we can achieve such a running time by a randomized algorithm)?

\bibliography{references}

\clearpage
\appendix
\appendixText

\end{document}